\documentclass{amsart}
\usepackage[utf8]{inputenc}
\usepackage[margin=1in]{geometry}
\usepackage{listings}
\usepackage{longtable}
\usepackage[ruled,noline]{algorithm2e}
\usepackage{amsmath}
\usepackage{amsfonts}
\usepackage{tikz}
\usepackage{subcaption}
\usepackage{amsthm}

\usetikzlibrary{intersections, calc, angles}

\ExplSyntaxOn
\NewDocumentCommand{\setlisttikz}{mm}
{
	\clist_clear_new:c { l_jens_#1_array_clist }
	\clist_set:cn { l_jens_#1_array_clist } { #2 }
}
\NewExpandableDocumentCommand{\listitem}{mm}
{
	\clist_item:cn { l_jens_#1_array_clist } { #2 }
}
\ExplSyntaxOff

\usepackage[shortlabels]{enumitem}

\makeatletter
\newtheorem*{rep@theorem}{\rep@title}
\newcommand{\newreptheorem}[2]{%
	\newenvironment{rep#1}[1]{%
		\def\rep@title{#2 \ref{##1}}%
		\begin{rep@theorem}}%
		{\end{rep@theorem}}}
\makeatother

\usepackage[sort]{natbib}
\usepackage{hyperref}
\usepackage{mathtools}

\mathtoolsset{showonlyrefs}

\DeclareMathOperator*{\argmin}{arg\,min}

\newcommand{\rr}{\mathbb{R}}

\newcommand{\nn}{\mathbb{N}}
\newcommand{\betab}{\boldsymbol{\beta}}

\newcommand{\psib}{\boldsymbol{\psi}}
\newcommand{\lambdab}{\boldsymbol{\lambda}}

\newcommand{\eg}{F_{\tree_v}}
\newcommand{\ec}{F_{\gamma}}
\newcommand{\er}{F_{\tree}}

\newcommand{\jlc}{J_{\gamma_L,\gamma_C}}
\newcommand{\jcg}{J_{\gamma_C,v}}
\newcommand{\jlg}{J_{\gamma_L,v}}

\newcommand{\jrc}{J_{\tau_L,\gamma_C}}
\newcommand{\jrg}{J_{\tau_L,v}}
\newcommand{\jrl}{J_{\tau_L,\gamma_L}}

\newcommand{\ub}{\boldsymbol{u}}
\newcommand{\yb}{\boldsymbol{y}}
\newcommand{\qb}{\boldsymbol{q}}

\newcommand{\ab}{\boldsymbol{a}}
\newcommand{\bb}{\boldsymbol{b}}
\newcommand{\stack}{\textrm{Stack}}
\newcommand{\tree}{\mathcal{G}}
\newcommand{\var}{\textrm{Var}}
\newcommand{\cov}{\textrm{Cov}}
\newcommand{\children}{\textrm{Child}}
\newcommand{\level}{\textrm{Level}}
\newcommand{\card}{\textrm{Card}}

\newreptheorem{theorem}{Theorem}
\newtheorem{theorem}{Theorem}
\newtheorem{lemma}{Lemma}

\theoremstyle{definition}

\newtheorem{remark}{Remark}
\newtheorem{example}{Example}

\date{\today}
\title{Least Squares Estimation For Hierarchical Data}
\author{Ryan Cumings-Menon}
\author{Pavel Zhuravlev}

\thanks{We thank the Editor, the Associate Editor, two anonymous referees, Aref Dajani, Sourya Dey, Justin Doty, Mark Fleischer, Caleb Floyd, Daniel Kifer, Philip Leclerc, Mary Pritts, and Rolando Rodr\'{i}guez for their helpful comments and suggestions. The views expressed in this paper are those of the author and not those of the U.S. Census Bureau. The Census Bureau has reviewed this data product to ensure appropriate access, use, and disclosure avoidance protection of the confidential source data (Project No. 7502798, Disclosure Review Board (DRB) approval number: CBDRB-FY24-CED005-0002). Works created by U.S. Government employees are not subject to copyright in the United States, pursuant to 17 U.S.C. § 105.}

\begin{document}

	\begin{abstract}
		The U.S. Census Bureau's 2020 Disclosure Avoidance System (DAS) bases its output on noisy measurements, which are population tabulations added to realizations of mean-zero random variables. These noisy measurements are observed for a set of hierarchical geographic levels, \textit{e.g.}, the U.S. as a whole, states, counties, census tracts, and census blocks. The Census Bureau released the noisy measurements generated in the DAS executions for the two primary 2020 Census data products, in part to allow data users to assess uncertainty in 2020 Census tabulations introduced by disclosure avoidance. This paper describes an algorithm that can leverage the hierarchical structure of the input data in order to compute very high dimensional least squares estimates in a computationally efficient manner. Afterward, we show that this algorithm's output is equal to the generalized least squares estimator, describe how to find the variance of linear functions of this estimator, and provide a numerical experiment in which we compute confidence intervals of tabulations based on this estimator. We also describe an accompanying Census Bureau experimental data product that applies this estimator to the publicly available noisy measurements to provide data users with the inputs required to derive confidence intervals for all tabulations that were included in the 2020 Redistricting Data File, for the U.S., state, county, and census tract geographic levels. 
	\end{abstract}
	
	\maketitle
	
	\section{Introduction}

	As described by \cite{abowd20222020} and \cite{cumings2023disclosure} in more detail, the U.S. Census Bureau's 2020 Disclosure Avoidance System (DAS) uses formally private methods \citep{dwork2006calibrating, bun2016concentrated} to protect the confidentiality of 2020 Census respondents. One important advantage of formally private methods is that they allow for a greater level of transparency relative to more classical disclosure limitation methods. For example, the Census Bureau has released the noisy measurements that were used to create the tabulations found in the published 2020 Census statistical data products \citep{DHCddps}; these are defined as demographic cross tabulations within specific geographic regions (\textit{e.g.}, the count of respondents that identify as both Black and Asian in Rhode Island, the count of respondents that identify as Hispanic or Latino in the U.S. as a whole, \textit{etc.}) added to realizations of mean-zero random variables. Since the distribution of each noisy measurement is also public, these noisy measurements allow users to compute confidence intervals (CIs) for the confidential cross tabulations, \textit{i.e.}, the cross tabulations prior to the application of disclosure limitation methods.\footnote{Note that we do not attempt to account for uncertainty in the underlying confidential counts of the decennial census in this paper, such as non-response and enumeration errors; instead, we view these confidential counts prior to the application of disclosure limitation methods as unknown and non-random population parameters.} URLs for these noisy measurements are available in Table 15 of \citep{cumings2023disclosure}.

	Many possible CIs can be derived by using only a small subset of the full set of noisy measurements. For example, one can use only the noisy measurement for the Asian population in Rhode Island (RI) to derive a CI for this same population count. However, it is also possible to improve on these na\"ive estimators by leveraging more information from the noisy measurements. In particular, some additional notation will be helpful to describe the unique estimator that leverages all of the noisy measurements, along with the sense in which this estimator can be viewed as optimal. Specifically, let $\betab\in \nn^N$ denote the vector of confidential histogram counts, \textit{i.e.}, the flattened fully-saturated contingency table, and let $\ub \in \rr^M $ denote the vector of independent and mean-zero random variables added to the confidential tabulations when defining the noisy measurements used by DAS.\footnote{While the counts based on the output of DAS can be biased in some cases and are generally dependent on one another, the random variables generated internally in DAS, \textit{i.e.}, the vector $\ub\in\rr^M$ in \eqref{eq:yb_def_intro}, are each mean-zero and independent random variables. In other words, the post-processing methods used within DAS take the unbiased and independent noisy measurements as input, but these properties are sacrificed by the post-processing methods used by DAS in order to satisfy other desiderata, such as consistency between the tabulations in different geographic levels and nonnegativity.} Thus, the vector of noisy measurements for either of the two main 2020 Census data products, \textit{i.e.}, the Redistricting Data File or the Demographic and Housing Characteristics File (DHC), is given by
	\begin{align}
		\yb &= S \betab + \ub,  \label{eq:yb_def_intro}
	\end{align}
	where $S\in\rr^{M\times N}$ has full column rank. Using this notation, the generalized least squares (GLS) estimator and its variance matrix are
	\begin{align}
		\tilde{\betab} &= \argmin_{\betab} (S \betab - \yb)^\top \var(\ub)^{-1} (S \betab - \yb)  = \left(S^\top \var(\ub)^{-1} S\right)^{-1} S^\top \var(\ub)^{-1} \yb  \label{eq:est_betab_naive}  \\
		\var(\tilde{\betab}) &=\left(S^\top \var(\ub)^{-1} S\right)^{-1}. \label{eq:var_betab_naive}
	\end{align}
	Aitken's Theorem states that this estimator is the best linear unbiased estimator (BLUE), which implies that for any alternative unbiased estimator for $\betab$ of the form $\check{\betab} = A\yb,$ we have $\var(\tilde{\betab}) \leq \var(\check{\betab}),$ in the sense that $ \var(\check{\betab})- \var(\tilde{\betab})$ is positive semidefinite \citep{aitken1935iv}.

	At this point, asymptotic normality of this estimator (\textit{i.e.}, in the limit as the number of noisy measurements, $M,$ diverges) could be used to derive a full-information CI of the unobserved value $\qb^\top \betab,$ for any given user-defined vector $\qb \in\rr^{N},$ using both $\qb^\top \tilde{\betab}$ and its variance. However, for the use case we have in mind, it is not feasible to compute this estimator directly because the dimension of both $\betab$ and $\yb$ (\textit{i.e.}, $M$ and $N,$ respectively) is over 10 billion for the 2020 Redistricting Data File, which is the smaller of the two main decennial census data products. 

	The purpose of this paper is to describe a computationally efficient method for computing $\qb^\top \tilde{\betab}$ and its variance to derive full-information CIs for counts of respondents in both arbitrary demographic groups and any geographic region. While this estimator uses the same noisy measurements generated in the 2020 production DAS executions as input, note that this estimator will not generally be equal to that of the 2020 Census tabulation, since DAS uses different post-processing steps to ensure additional constraints hold, such as non-negativity constraints. Also, note that we view the population parameter $\betab$ as fixed throughout the paper, and our main goal is compute CIs on this population parameter to account for uncertainty that is due solely to the application of statistical disclosure limitation methods, rather than other sources of uncertainty in decennial census counts, such as non-response and enumeration errors. The main property that allows for computational feasibility is that the geographic units used to define these noisy measurements are defined hierarchically in a rooted tree, \textit{e.g.}, the vertex corresponding to the U.S. as a whole splits into states, which are in turn further divided into vertices corresponding to more granular geographies, all the way down to the most granular set of vertices, census blocks, as illustrated in the next section in Figure \ref{fig:spine}.

	While our motivating use case concerns 2020 Census data, we expect the methods described here to be of interest in other settings in which the input data are also hierarchical. As described in Section \ref{assumptions} in more detail, the primary requirements for our proposed approach, in addition to the standard Gauss-Markov assumptions, include: 1) The data must be hierarchical, in the sense that there exists a rooted tree such that the unknown population parameter (\textit{e.g.}, a vector of population counts) associated with each non-leaf vertex is equal to the sum of those of its child vertices; and 2) The observations associated with each vertex in this tree must be independent of the observations of every other vertex. The first of these two assumptions can be viewed as a requirement that we place on the matrix $S.$ Example \ref{eg:intro} describes our requirement regarding the unobserved coefficients being defined in a hierarchical manner in the case of a particularly simple tree, along with the implication of this assumption on the structure of $S.$ Note that this example denotes a set of column vectors, $\{\boldsymbol{q}_i\}_i,$ concatenated together as $\stack(\{\boldsymbol{q}_i\}_i),$ which we also use in the remainder of the paper.

	\begin{example} \label{eg:intro}
		We will encode the hierarchy of units for which we observe noisy measurements using the rooted tree denoted by $\tree.$ Suppose $\tree$ is composed of the root vertex $r$ and that this vertex has the child vertices $c$ and $d.$ Suppose that we observe both a total population and a Voting\_Age (\textit{i.e.}, both the population count for those that are below 18 years old and the count of those that are at least 18) noisy measurement for each of these three vertices. In other words, if we denote the confidential (\textit{i.e.}, unobserved) Voting\_Age counts for each vertex $v\in \{r, c,d\}$ by $\betab(v)\in\nn^2,$ then the noisy measurements we observe for vertex $v$ are defined as
		\begin{align} \label{vertex_specific_nmfs}
			\yb(v) = S(v) \betab(v) + \ub(v)
		\end{align}
		where the vertex-specific design matrix $S(v)\in \rr^{m\times n}$ is defined as 
		\[S(v) = \begin{bmatrix}
			1 & 1 \\
			1 & 0 \\
			0 & 1
		\end{bmatrix},\]
		which encodes the total population query in the first row and both levels of the Voting\_Age query in the final two rows, and where each element of $\ub(v)\in \rr^3$ is distributed as $\ub(v)[i]\sim \textrm{iid}(0, \sigma^2_i);$ note that the specific distribution used in our primary use case is described in Section \ref{sec:setting}.

		Using this notation, the vectorized noisy measurements for all three vertices can be defined in a similar way as \eqref{eq:yb_def_intro}; specifically, 
		\begin{align*}
			\yb &=\stack(\{\yb(r),\yb(c),\yb(d)\}) = S \betab + \ub \\
            &=  \begin{bmatrix}
				S(r) & S(r) \\
				S(c) & 0 \\
				0 & S(d)
			\end{bmatrix}  \stack(\{\betab(c),\betab(d)\}) + \stack(\{\ub(r), \ub(c), \ub(d)\}) \\
			&= \begin{bmatrix}
				1 & 1 & 1 & 1 \\
				1 & 0 & 1 & 0\\
				0 & 1 & 0 & 1 \\
				1 & 1 & 0 & 0 \\
				1 & 0 & 0 & 0\\
				0 & 1 & 0 & 0\\
				0 & 0 & 1 & 1 \\
				0 & 0 & 1 & 0 \\
				0 & 0 & 0 & 1
			\end{bmatrix} \stack(\{\betab(c),\betab(d)\}) + \stack(\{\ub(r), \ub(c), \ub(d)\}). 
		\end{align*}

        Our requirement that the input data are hierarchical amounts to requiring that the noisy measurements of each vertex $v\in \tree$ depend on the histograms of its set of descendant leaves (say, $D)$ in the same way, meaning $\yb(v) = S(v)\left(\sum_{d\in D} \betab(d)\right) +\ub(v).$ In the context of this example, this ensures we can define the noisy measurements using the individual vertex definition given in \eqref{vertex_specific_nmfs} or a definition that defines all noisy measurements simultaneously, as in \eqref{eq:yb_def_intro}.  \qed
	\end{example}

	An implementation of our proposed methods for our primary use case is also available at \url{https://github.com/uscensusbureau/DAS_2020_GLS_Uncertainty_Evaluation}; as described in Section \ref{sec:upcoming_ci_release} in more detail, this codebase was used for a recent experimental data product that provides the required inputs to generate CIs for all published 2020 Redistricting Data File tabulations down to the census tract geographic level.

	After introducing notation in the next subsection, the remainder of this section describes our primary motivating use case in more detail, and reviews the previous literature on methods related to the ones we describe here. Both of our main results on properties of the proposed estimator, which, respectively, show that our proposed algorithm outputs the GLS estimator and describe a case in which inference is exact, are provided in Section \ref{sec:main_results} and proved in the appendices. Section \ref{sec:two_pass_est} describes our proposed approach for estimating the full-information generalized least squares (GLS) estimator, which we call the Two-Pass Algorithm because this algorithm performs operations recursively from the leaf vertices up to the root vertex and then from the root vertex down to the leaves. Section \ref{sec:cov} describes a computationally efficient way to compute the covariance between the GLS estimates of arbitrary pairs of vertices of the tree. Section \ref{sec:summary} describes how the proposed algorithms can be used to compute CIs for user-defined queries in arbitrary geographic regions. Section \ref{sec:numerical_exp} provides a numerical experiment of our proposed approach, to explore the accuracy of the resulting CI estimates and the computational requirements of our proposed algorithms. In Section \ref{sec:connection_w_hays_et_al} we show that our proposed method is a generalization of the estimator described by \cite{hay2010boosting}. Our proposed algorithms were also used to create a recent statistical data product release, which is described in Section \ref{sec:upcoming_ci_release}, and Section \ref{sec:conclusion} concludes.

	\subsection{Notation} \label{sec:notation}

	Throughout the paper we denote matrices using uppercase, and vectors using bold lowercase. Given the matrices $A,B\in\rr^{N\times N}$ we use $A\leq B$ to denote the property that all eigenvalues of $B-A$ are real and the minimum eigenvalue of $B-A$ is non-negative. Also, let $A \otimes B$ denote the Kronecker product of $A,B\in \rr^{M\times N}.$ Given the random vectors $\boldsymbol{a}\in\rr^M$ and $\boldsymbol{b}\in\rr^N,$ each with elements with finite variance, we use $\var(\boldsymbol{a})\in\rr^{M\times M}$ to denote the variance matrix of $\boldsymbol{a}$ and $\cov(\boldsymbol{a}, \boldsymbol{b})\in\rr^{M\times N}$ to denote the covariance matrix between $\boldsymbol{a}$ and $\boldsymbol{b}.$ The length $N$ column vector with each element equal to $k\in\rr$ is denoted by $\boldsymbol{k}_N,$ and when there is little risk of confusion, we omit the subscript and write $\boldsymbol{k}$ instead. We denote the $N\times N$ identity matrix by $I_N,$ the $i^\textrm{th}$ row of $A\in \rr^{M\times N}$ by $A[i,\cdot],$ and the $j^\textrm{th}$ column of $A$ by $A[\cdot,j].$ We also denote the cardinality of the finite set $S$ by $\card(S).$

	Since we primarily use terminology from graph theory (rather than the terminology used by the Census Bureau), it will be helpful to describe this notation here. Specifically, let the rooted tree be denoted by $\tree.$ To denote the subtree of $\tree$ rooted at vertex $v\in \tree,$ \textit{i.e.}, the subset of $v$ that includes $v\in\tree$ and all of its descendants, we use $\tree_v.$ We assume that $\tree$ is defined so that the shortest path from any given leaf vertex to the root vertex all have a length equal to $L,$ and for any level index $l\in\{0,\dots, L\},$ we use $\level(\tree, l)$ to denote the set of vertices in level $l.$ We also use $\children(v)$ to denote the set of child vertices of $v\in\tree.$

	For each vertex $v\in\tree,$ we will denote the vector of (unknown) independent variables associated with the vertex by $\betab(v)\in\rr^n.$ Throughout the paper, we assume that parent-child consistency holds; in other words, for any vertex $v\in\tree$ that is not a leaf, we assume 
	
	\begin{align}
		\sum_{c\in\children(v)} \betab(c) = \betab(v). \label{eq:par_child_cons}
	\end{align}
	
	\noindent
	We also associate with each vertex $v\in \tree$ a full column rank design matrix $S(v)\in \rr^{m\times n},$ and a vector of observations defined by $\yb(v) = S(v) \betab(v) + \ub(v),$ where $\ub(v)\in\rr^{m}$ is a mean-zero random variable with finite variance. Note that we do not assume that the design matrix $S(v),$ or the distribution of $\ub(v),$ is the same for each vertex $v\in \tree.$\footnote{One implication of our use of $m$ and $n$ to denote the number of rows and columns of $S(v)$ throughout the paper is that this implicitly requires that these dimensions are the same for each vertex $v\in\tree;$ however, we actually only require the number of columns of $S(v),$ \textit{i.e.}, $n\in\nn,$ to be the same for each vertex.} In the context of our main use case, $\ub(v)[i]$ is independent of the remaining elements of $\ub(v),$ and also independent of all elements of $\ub(c)$ for each vertex $c\in\tree$ that is not equal to $v.$ However, our proposed approach actually only requires the latter of these two conditions.  

    We perform various stacking operations on the attributes of the vertices below, \textit{e.g.}, $\stack(\{S(v)\betab(v)\}_{v\in\tree}).$ These operations require that all sets of vertices that we use in the paper are totally ordered to ensure that the vertex ordering used within these stack operations are consistent with one another, \textit{e.g.}, so that $\stack(\{\yb(v) - S(v)\betab(v)\}_{v\in\tree}) = \stack(\{\ub(v)\}_{v\in\tree}).$

	\subsection{The Decennial Census Setting} \label{sec:setting}
	
	The official documentation provided by the Census Bureau often uses terms that differ from those that are commonly used in graph theory. This section briefly outline a few of these alternative terms to describe how the notation introduced in the previous subsection relates to our primary motivating use case. Specifically, the Census Bureau's concept of a \textit{geographic spine} can alternatively be viewed as a rooted tree. The spines containing the primary geographic units for which decennial census tabulations are released are known as the U.S. or Puerto Rico (PR) \textit{tabulation spines}, which have root vertices corresponding to, respectively, the U.S. as a whole (\textit{i.e.}, the region defined as the union of all 50 states and Washington, DC) and PR. The geographic extent of each vertex in any spine corresponds to an element of a partition of the geographic extent of the root vertex. The collection of vertices corresponding to elements within the same partition is called a \textit{geographic level}. For example, in the tabulation U.S. spine, the U.S. geographic level contains only one vertex, corresponding to the geographic extent of the U.S. as a whole, and the state level consists of the children of the U.S. vertex, \textit{i.e.}, 51 units, each corresponding to either one of the 50 states or Washington DC. The other geographic levels on the tabulation U.S. spine are the county, census tract, block group, and block levels. Figure \ref{fig:spine} provides a graphical representation of the tabulation U.S. spine, including an example of a path from the U.S. geographic unit to a census block. 
	
	The data product accompanying this paper contains the information required to create CIs for geographic units on these tabulation spines. However, the noisy measurements that were leveraged to create these estimates were defined using alternative geographic spines, which were called optimized spines. For the purposes of this paper, we do not require more detail on how DAS defines the optimized spine internally prior to generating the noisy measurements, aside from the fact that this internal spine is distinct from the one for which we provide count estimates and CIs, \textit{i.e.}, the tabulation geographies.\footnote{More detail on why optimized spines are used internally in DAS, rather than the tabulation spine directly, as well as how the optimized spines are defined, can be found in \citep{cumings2024geographic}.}

	% As their names suggest, the U.S. and PR ``tabulation" spines contain many of the geographic units for which the Census Bureau publishes tabulations as part of various decennial census statistical data products. While we are primarily interested in computing unbiased estimates and CIs for geographic units in these tabulation spines in our main use case, the noisy measurements that our proposed estimator will use as input contain noisy measurements for geographic units in a different set of spines called \textit{optimized spines}. This paper will not require more detail on the specific way in which these optimized spines that are used internally within DAS executions are defined, other than the fact that the geographies for which we aim to provide count estimates and CIs, \textit{i.e.}, the tabulation geographies, are distinct from the geographies for which we observe noisy measurements.\footnote{For more detail on why we use distinct spines within DAS and the methods we use to define these spines, see \citep{cumings2024geographic}.}

	\begin{figure}
		\centering
		\begin{tikzpicture}[
			level 1/.style = {sibling distance = 1.7cm}, 
			level 2/.style = {sibling distance = 1.7cm},
			level 3/.style = {sibling distance = 1.7cm},
			level 4/.style = {sibling distance = 1.7cm},
			level 5/.style = {sibling distance = 1.7cm},
			]
			\setlisttikz{geographic levels}{U.S.,State,County,Tract,Block Group,Block}
			\node (root) {U.S.}
			child {
				node (alabama) {Alabama (01)}
				child {
					node (autauga) {Autauga (01001)}
					child {
						node (tract1) {01001020100}
						child {
							node (bg1) {010010201001}
							child {
								node (b1) {0100102010011000}
							}
							child {node (b2) { }}
							child {node (b3) { }}
						}
						child {node (bg2) { }}
						child {node (bg3) { }}
					}
					child {node (tract2) { }}
					child {node (tract3) { }}
				}
				child {node (county2) { }}
				child {node (county3) { }}
			}
			child {node (state2) { }}
			child {node (state3) { }};
			
			\draw[dotted, very thick,<->] (b2) to [bend right=15] (b3);
			\draw[dotted, very thick,<->] (bg2) to [bend right=15] (bg3);
			\draw[dotted, very thick,<->] (tract2) to [bend right=15] (tract3);
			\draw[dotted, very thick,<->] (state2) to [bend right=15] (state3);
			\draw[dotted, very thick,<->] (county2) to [bend right=15] (county3);
			
			\coordinate (topRight) at ($(current bounding box.north east) + (2,0)$);
			
			\foreach \y  in {0, 1, 2, 3, 4, 5} {
				\node[label] at ($(topRight) + (0, -\y * 1.505 -.26)$) {\listitem{geographic levels}{\y+1} Geolevel $(l=\y)$};
			};
		\end{tikzpicture}
		\caption{This is a graphical depiction of the 2020 tabulation spine. The right column provides the geographic level names and indices. The left side of the figure provides an example of a path between the root and a leaf vertex. The census geographic codes for the geographic units on this path are provided in parentheses.}
		\label{fig:spine}
	\end{figure}

	Both the numerical experiment described in Section \ref{sec:numerical_exp} and the accompanying data product described in Section \ref{sec:upcoming_ci_release}, leverage noisy measurements defined in the histogram schema of the persons universe of the 2020 Redistricting Data File. For each geographic unit, these tabulations are composed of cross products of the following attributes, as described by \cite{abowd20222020} in more detail.

	\begin{itemize}
		\item Household or Group Quarters Type (HHGQ), 8 levels: Provides counts of those living in households, correctional facilities for adults, juvenile facilities, nursing facilities/skilled-nursing facilities, other institutional facilities, college/university student housing, military quarters, and those living in other non-institutional facilities
		\item Voting Age (VOTING\_AGE), 2 levels: Provides counts of individuals that are age 17 or younger and those that are age 18 or older
		\item Hispanic or Latino Origin (HISPANIC), 2 levels: Provides counts of individuals that are Hispanic/Latino and those that are not Hispanic or Latino
		\item Census Race (CENRACE), 63 levels: Provides counts of individuals that identify as each combination of Black/African American, American Indian/Native Alaskan, Asian, Native Hawaiian/Pacific Islander, White, and some other race, except ``none of the above."
	\end{itemize}

	As described above, the tabulations used to define each noisy measurement in the noisy measurements is defined as the cross product of subsets of one or more of these attributes. After each confidential scalar element of the tabulation is computed, the final noisy measurement is defined by adding a mean-zero independent random variable to this scalar. Specifically, these random variables each follow a \textit{discrete Gaussian} distribution \citep{canonne2020discrete}, which has a probability mass function given by
	\[f(x) \propto \exp\left(-x^2 / (2 \sigma^2) \right),\]
	but, unlike the Gaussian distribution, it has a support given by the integers. The parameter $\sigma^2$ for each of these random variables depends on both the vertex and tabulation, and was set to be consistent with a $\rho$-zCDP privacy accounting framework \citep{bun2016concentrated}; for more detail, see \cite{abowd20222020}.

	Note that \cite{DHCddps} provides instructions for downloading the production 2020 Redistricting Data File and DHC noisy measurements. The DHC noisy measurements use a larger schema (\textit{i.e.}, one that is strictly more granular) than the Redistricting Data File schema to support the estimation of more detailed tabulations. For example, while the persons universe Redistricting Data File schema described above consists of $n = 2,016 = 8 \times 2 \times 2 \times 63$ detailed histogram cells in each vertex, the persons universe DHC schema consists of 1,227,744 detailed cells in each vertex.

	\subsection{Previous Literature} \label{sec:previous_literature}
	
	Our proposed estimation algorithm is a generalization of the approach described by \cite{hay2010boosting}, in the sense that these two approaches provide identical estimates when $m=n=1,$ $\var(\ub(v))$ is the same for each vertex $v\in\tree,$ and the number of children of each parent vertex is constant. More detail on this connection is provided in Section \ref{sec:connection_w_hays_et_al}. Note that, in addition to providing an algorithm for the GLS estimate itself, we also provide a method that computes the exact variance of the estimate, rather than a bound on this variance. \cite{honaker2015efficient} builds on \citep{hay2010boosting} by providing additional analysis and a simulation exercise. Also, \cite{xu2013differentially} describes variants of the algorithm proposed by \cite{hay2010boosting}, including those that allow for parent vertices to have differing numbers of children and for these numbers of children to be chosen adaptively based on the data. In the differential privacy literature, mechanisms that add noise to linear queries, as is done in \eqref{eq:yb_def_intro}, and then return a GLS estimator based on the observations are examples of matrix mechanisms, which was originally described by \cite{li2010optimizing}.

	There are also a few similarities between the approach described here and \textit{forward-backward} algorithms described in the literature on hidden Markov models on trees; see for example, \citep{willsky2002multiresolution}. However, these approaches are not directly related to the one proposed here because they assume the data generating process satisfies a Markov property, which is precluded in our setting by our parent-child consistency assumption provided in \eqref{eq:par_child_cons}. 

    Other work has been done to compensate for various types of errors in the statistical data products published by the Census Bureau. In particular, \cite{agarwal2021causal} provide methods to compensate for a much more general class of errors than the ones we consider here, including non-response and enumeration errors, without strong assumptions on the parametric distribution of the errors used for the application of statistical disclosure limitation methods. As described previously, one of the main advantages of the Census Bureau's adoption of formally private methods is that the methods themselves can be transparently communicated, which is not possible for the disclosure limitation methods that were used by the Census Bureau for previous decennial censuses, including the distributions of the noisy measurements that are generated within the DAS, \textit{i.e.}, $\yb$ in \eqref{eq:yb_def_intro}, the algorithms that post-processes these noisy measurements, the code repositories in which these algorithms were implemented, and the realizations of the noisy measurements themselves \citep{abowd20222020,cumings2023disclosure}. Our proposed approach differs from that of \cite{agarwal2021causal} because we only characterize the uncertainty due to disclosure avoidance, since, unlike other sources of error, the publicly available noisy measurements can be leveraged to accurately model these errors using a closed-form distribution. However, in the context of the decennial census, errors from other sources, including non-response and enumeration errors, are often more significant than those due to the application of disclosure limitation methods. Other work has focused on characterizing these other errors and modeling their magnitudes in the absence of uncertainty due to the application of disclosure limitation methods. For example, coverage estimates are provided in U.S. Census Bureau reports based on the 2020 Post-Enumeration Survey (PES), such as \citep{khubba2022national}. Also, \cite{schafer2021} describe two approaches to leverage the PES to create conservative bounds on errors unrelated to disclosure avoidance in the context of the 2010 Census.

    There are also other feasible approaches to solve large least-squares problems using iterative approaches that leverage sparsity in design matrices. For example, the least-squares estimator computed by our approach could alternatively be computed using the stochastic gradient descent algorithm or the MINRES algorithm proposed by \cite{paige1975solution}. In contrast to these alternatives, our proposed approach also provides the variance of arbitrary linear products with the least-squares estimator, which is required to construct confidence intervals in our primary use case. Since our proposed approach uses a direct (rather than iterative) algorithm, our approach also avoids possible concerns related to convergence and numerical tolerances.

	\subsection{Assumptions} \label{assumptions}

	We use several interrelated assumptions throughout the paper, most of which were already described above. These are also summarized in the table below for later reference.

	\begin{enumerate}[label=(\roman*)]
		\item\label{assump:spine_struct} $\tree$ is a rooted tree, and, for each leaf vertex $v\in\tree,$ there are $L$ edges between $v$ and the root of $\tree.$ Also, for each parent vertex $v,$ \eqref{eq:par_child_cons} holds, \textit{i.e.}, $\sum_{c\in\children(v)} \betab(c) = \betab(v).$ 
		\item\label{assump:srank} For each $v\in\tree,$ the design matrix $S(v)\in \rr^{m(v) \times n}$ is fixed and has full column rank.
		\item\label{assump:u_iid} For each $v\in\tree,$ $\yb(v)=S(v) \betab(v) + \ub(v),$ where $\ub(v)$ is a mean zero random vector with variance matrix $\var(\ub(v))$ and, for each $w\in\tree$ not equal to $v,$ $\ub(w)$ is independent of $\ub(v).$
		\item\label{assump:normality} For each $v\in\tree,$ $\ub(v)\sim N(\boldsymbol{0}, \var(\ub(v))).$ 
	\end{enumerate}

	We assume these conditions hold throughout the paper with the exception of Assumption \ref{assump:normality}, which is only needed to prove that our proposed method provides exact inference in finite samples. Assumption \ref{assump:spine_struct} can be understood as the requirement that the dataset be in a hierarchical format. Note that spines that do not have $L$ edges between the root and each leaf vertex can still be used in this framework after adding vertices between each leaf and its parent vertex, so this assumption is primarily for notational convenience. The added vertices will not impact the final estimates as long as their $\var(\ub(v))$ is defined as a diagonal matrix with diagonal entries equal to infinity.

	Assumptions \ref{assump:srank} and \ref{assump:u_iid} are standard in work related to linear regressions, since they are a variant of the assumptions required for the Aitken's Theorem, which is closely related to the Gauss-Markov Theorem and is provided for completeness in Lemma \ref{lem:appendix:gauss_markov_thm} \citep{aitken1935iv}. Note that it is possible to prove our main result even after relaxing \ref{assump:srank} by removing the assumption that $S(v)$ is fixed for each $v\in\tree,$ and instead only assuming it is exogenous, \textit{i.e.}, $E(\ub | S) = 0;$ for more detail, see \citep{greene2003econometric}. 

    It is also worth pointing out that we assume $\var(\ub(v))$ is known for each $v\in\tree$ throughout our paper. This condition holds in our primary use case, since internally DAS draws each error vector $\ub(v)$ from a predefined distribution, as described in Section \ref{sec:setting} in more detail. However, there are also multiple ways to use our proposed approach in cases in which this exact variance matrix is unknown that would still be consistent. First, one option would be to simply use the modeling assumption that the errors are homoscedastic, \textit{i.e.}, $\var(\ub(v))=I\sigma^2.$ Second, one could also use an approach that is analogous to feasible GLS \citep{greene2003econometric}, but some care must be taken when using this approach to ensure the initial variance error estimates imply the errors are uncorrelated between each pair of vertices $w,v\in\tree,$ as described in Assumption \ref{assump:u_iid}. This can be done for each $v\in\tree$ by estimating the initial variance matrix $\var(\ub(v))$ using only attributes associates with vertex $v,$ such as $S(v)$ and $\yb(v).$ Consistency of this approach would require the number of observations associated with each vertex to diverge.

	\subsection{Main Results} \label{sec:main_results}

	The first main result of this paper shows that the algorithm proposed here is a generalization of that of the one proposed by \cite{hay2010boosting}, which is stated below and proved in Section \ref{sec:connection_w_hays_et_al}. 

	\begin{reptheorem}{thm:hay_et_al_equality}
		Suppose Assumptions \ref{assump:spine_struct}-\ref{assump:u_iid} hold and that, for every $v\in\tree,$ $S(v)$ is equal to the matrix $\begin{bmatrix} 1 \end{bmatrix}$ and $\var(\ub(v))$ is equal to the same scalar value. Also, suppose that each parent vertex $v\in\tree$ has exactly $k\in\nn$ children. Then the estimator described by \cite{hay2010boosting} is identical to the estimator $\{\tilde{\betab}(v)\}_{v\in\tree}$ proposed in Section \ref{sec:two_pass_est}.
	\end{reptheorem}

	As described above, we also show that the estimator provided by our approach is the GLS estimator, which is stated below and proved in Appendix \ref{sec:blue_proof}.

	\begin{reptheorem}{thm:appendix:beta_tilde_is_blue}
		If Assumptions \ref{assump:spine_struct}-\ref{assump:u_iid} hold, then for each $v\in\tree$  $\tilde{\betab}(v),$ as defined in Section \ref{sec:two_pass_est}, and the value of $\tilde{\betab}_{H,\qb}$ returned from Algorithm \ref{alg:est_ci} is the full-information BLUE for $\betab(v).$ 
	\end{reptheorem}

	Before describing the next result, it may be helpful to restate that we view the population parameter $\betab$ in \eqref{eq:yb_def_intro} as fixed throughout the paper, and our main goal is compute confidence intervals on this population parameter, particularly for cases in which there are a large number of observations, \textit{i.e.}, $\yb$ in \eqref{eq:yb_def_intro}. In particular, since we essentially condition on $\betab$ in this paper, we do not derive inferences on errors that are unrelated to the application of statistical disclosure limitation methods, even though these other sources of error can be significant, because the statistical disclosure limitation methods used for the 2020 Census are transparent enough so that this uncertainty can be characterized without requiring strong modeling assumptions.

    With this in mind, the next main result is stated below and proved in Appendix \ref{sec:appendix:beta_tilde_is_normal}. Unlike the previous two results, a normality assumption is required for this result to hold in finite samples, \textit{i.e.}, when the dimension of $\yb$ in \eqref{eq:yb_def_intro} is fixed. As described above, for the primary use case we have in mind, $\ub(v)$ is distributed as a \textit{discrete} Gaussian, so this assumption does not quite hold in this use case. We provide a numerical experiment in Section \ref{sec:numerical_exp} to verify that the proposed approach appears to work well for the use cases we have in mind even when this assumption does not hold. In practice we expect our proposed approach will typically provide CIs with good coverage values in other use cases as well, even when Assumption \ref{assump:normality} does not hold, as long as there are a reasonably large number of observations, as is also described in Appendix \ref{sec:appendix:beta_tilde_is_normal} in more detail.

	\begin{reptheorem}{thm:appendix:beta_tilde_is_normal}
		If Assumptions \ref{assump:spine_struct}-\ref{assump:normality} hold, then $\tilde{\betab}_{H,\qb}$ is normally distributed, and the $1-\alpha$ CI of $\betab_{H,\qb}$ output from Algorithm \ref{alg:est_ci} is statistically valid.
	\end{reptheorem}

    The final result stated in this section provides the time complexity of our main algorithm; this result is proved in Appendix \ref{sec:appendix:time_complexity}. To compare this time complexity with that of the direct approach of computing $\tilde{\betab}$ using \eqref{eq:est_betab_naive}, if $M$ and $N$ are defined so that $S\in\rr^{M\times N},$ this direct approach has a time complexity of $O(M^2 N + N^3),$ which can be simplified to $O(M^2 N)$ because $M\geq N$ by Assumption \ref{assump:srank}. For the case in which a fixed proportion of the vertices of $\tree$ are leaves, then, using this notation and the assumptions introduced in the statement of the next theorem, the time complexity of this direct approach can be written as $O(m^2 n V^3)$ because $M$ and $N,$ respectively, are proportional to $m V$ and $n V.$ In other words, our proposed approach provides a time complexity that is lower than the direct approach by a factor of $V^2,$ so it is particularly advantageous for hierarchies that have a large number of vertices. This is the case for our primary use case because there are over 6 million leaf vertices in the U.S. geographic spine.

    \begin{reptheorem}{thm:appendix:time_complexity}
        Suppose that for each vertex $v\in\tree,$ $\yb(v) \in\rr^m$ and $\betab(v)\in \rr^n.$ Also, let $V$ be defined as the total number of vertices in $\tree,$ \textit{i.e.}, $V=\sum_l \card(\level(\tree, l)).$ Then the time complexity of Algorithm \ref{alg:two_pass_est} is $O(m^2 n V).$
	\end{reptheorem}

	\section{The Two-Pass Algorithm} \label{sec:two_pass_est}

	This section describes our proposed approach for computing the GLS estimator, which is provided in pseudocode in Algorithm \ref{alg:two_pass_est}. Table \ref{notation} provides a brief summary of the definitions used in this algorithm. The derivations of the formulas used in this algorithm are provided in Appendix \ref{sec:derivation_of_two_pass_algorithm}.

    At a high level, Algorithm \ref{alg:two_pass_est} performs a series of three updates on state variables within each vertex $v\in\tree,$ which include an estimate of $\betab(v)$ and its variance matrix. Each of the three updates to the state variable estimate of $\betab(v)$ defines this estimate as the GLS estimator that leverages progressively larger information sets, \textit{i.e.}, the subsets of observations used for estimation. Specifically, the first step of the algorithm is to define the estimate of each vertex $v$ as the GLS estimator based only on the observations associated with vertex $v$ itself. Second, in the fine-to-coarse recursion, this state variable is updated as the inverse-variance-weighted mean of this initial estimate and an estimate defined as the sum over the fine-to-coarse recursion estimates of the children of vertex $v,$ which is the GLS estimate for the sample defined as all observations associated with vertices in the subtree $\tree_v,$ as proved in Theorem \ref{thm:appendix:g_cond_g_minus_is_blue_for_any_level}. Similarly, in the coarse-to-fine recursion, which is the third and final update, this information set is expanded further to also include the observations of higher levels of $\tree$ by projecting the estimate from the fine-to-coarse recursion onto the set of estimates that satisfy parent-child consistency constraint \eqref{eq:par_child_cons}.

    \begin{table}[tbh]
		\centering
		\begin{tabular}{ l l }
			$\tree$ & The rooted tree \\
			$\tree_v$ & The subtree of $\tree$ rooted at vertex $v$ \\
			$\level(\tree, l)$ & The set of vertices that are $l\in\{0,1,\dots,L\}$ edges from the root vertex of $\tree$ \\
			$\children(v)$ & The set of vertices in $\tree$ that are children of $v$ \\
			$S(v) \in\rr^{m\times n}$ & The design matrix of vertex $v$ \\
			$\yb(v)\in\rr^m$ & The observations of vertex $v$ \\
			$\ub(v)\in\rr^m$ & The error component of the observations of vertex $v$ \\
			$\var (\boldsymbol{z})$ & The variance matrix of the random variable $\boldsymbol{z}$  \\
			$\betab(v)\in\rr^n$ & The (unobserved) vector of independent variables of vertex $v$ \\ 
			$\hat{\betab}(v | v)\in\rr^n$ & Estimate for vertex $v$ based on $\yb(v)$ \\ 
			$\hat{\betab}(v | v-)\in\rr^n$ & Estimate for vertex $v$ based on observations in $\tree_v$ \\ 
			$\tilde{\betab}(v)\in\rr^n$ & Estimate for vertex $v$ based on observations in all vertices in $\tree$ 
		\end{tabular}
		\caption{This table summarizes the definitions we use to compute the GLS estimator in Algorithm \ref{alg:two_pass_est}.} \label{notation}
	\end{table}

    \begin{algorithm}[tbh]
		\DontPrintSemicolon
		\SetKwInOut{Output}{return}
		\SetKwInOut{Input}{input}
		\Input{$\tree:$ The rooted tree, with the following objects associated with each vertex $v\in\tree :$ the vector of observations $\yb(v),$ the design matrix $S(v),$ and the variance matrix of the errors $V(\ub(v)).$}
		\tcp{Compute GLS estimates using only observations within each vertex:}
		$\hat{\betab}(v | v) \gets \left(S(v)^\top \var(\ub(v))^{-1} S(v)\right)^{-1} S(v)^\top \var(\ub(v))^{-1} \yb(v) $\;
		$\var(\hat{\betab}(v | v)) \gets \left(S(v)^\top \var(\ub(v))^{-1} S(v)\right)^{-1}$\;
		\tcp{Initialize fine-to-coarse recursion at leaf vertices:}
		\For{$v \in \level(\tree, L)$}{
			$\hat{\betab}(v | v-) \gets \hat{\betab}(v | v)$\;
			$\var(\hat{\betab}(v | v-)) \gets \var(\hat{\betab}(v | v))$}
		\tcp{Perform fine-to-coarse recursion:}
		\For{$l \in \{L-1, L-2, \dots, 0\}$}{
			\For{$v\in \level(\tree, l)$}{
				$\hat{\betab}(v | \children(v)-)  \gets \sum_{c\in\children(v)} \hat{\betab}(c| c-)$\;
				$\var(\hat{\betab}(v | \children(v)-))  \gets \sum_{c\in\children(v)} \var(\hat{\betab}(c| c-))$\;
				$\var(\hat{\betab}(v | v-)) \gets \left(\var(\hat{\betab}(v | v))^{-1} + \var(\hat{\betab}(v | \children(v)-))^{-1}\right)^{-1} $ \;
				$ \hat{\betab}(v | v-)  \gets \var(\hat{\betab}(v | v-)) \left(\var(\hat{\betab}(v | v))^{-1} \hat{\betab}(v | v) + \var(\hat{\betab}(v | \children(v)-))^{-1} \hat{\betab}(v | \children(v)-) \right)$}}
		\tcp{Initialize coarse-to-fine recursion at the root vertex, $r=\level(\tree, 0)[0]:$}
		$\tilde{\betab}(r) \gets \hat{\betab}(r | r-)$\;
		$\var(\tilde{\betab}(r)) \gets \var(\hat{\betab}(r | r-))$\;
		\tcp{Perform coarse-to-fine recursion:}
		\For{$l \in \{0, 1, \dots, L-1\}$}{
			\For{$v\in \level(\tree, l)$}{
				\For{$c\in\children(v)$}{
					$ A(c) \gets \var(\hat{\betab}(c | c-)) \left( \sum_{c'\in\children(v)} \var(\hat{\betab}(c' | c'-)) \right)^{-1}$\;
					$ \tilde{\betab}(c) \gets \hat{\betab}(c | c-) + A(c) \left(\tilde{\betab}(v) - \sum_{c'\in\children(v)} \hat{\betab}(c' | c'-)\right) $\;
					$\var(\tilde{\betab}(c)) \gets \var( \hat{\betab}(c | c-))   - A(c) \var( \hat{\betab}(c | c-))  + A(c)\var( \tilde{\betab}(v)) A(c)^\top$}}}
		\Output{$\{\tilde{\betab}(v), \var(\tilde{\betab}(v)), A(v),\var(\hat{\betab}(v  |  v-))\}_{v\in\tree}$}
		\caption{TwoPassGLS: Returns matrices and vectors that can be used in Algorithm \ref{alg:est_ci} to generate point estimates and CIs of user-defined queries for an arbitrary set of leaves} \label{alg:two_pass_est}
	\end{algorithm}

	A description of how this algorithm fits into the primary use case for this paper may also be helpful. Specifically, the final state variables $\{\tilde{\betab}(v)\}_{v\in\tree}$ and $\{\var(\tilde{\betab}(v))\}_{v\in\tree}$ are directly related to the standard GLS estimate, $\tilde{\betab},$ and its variance matrix, as defined in \eqref{eq:est_betab_naive}-\eqref{eq:var_betab_naive}, because $\tilde{\betab} = \stack(\{\tilde{\betab}(v)\}_{v\in\level(\tree,L)})$ and the diagonal blocks of $\var(\tilde{\betab})$ are given by $\{\var(\tilde{\betab}(v))\}_{v\in\level(\tree,L)}.$ No more information is required to estimate CIs that only depend on the estimate of a single vertex $v\in\tree,$ but, in more general cases, we also require $\cov(\tilde{\betab}(c),\tilde{\betab}(d))$ for arbitrary $c,d\in\tree.$ The next section describes how to compute these covariance matrices, and Section \ref{sec:summary} describes our proposed method to construct CIs using these estimates and covariance matrices. 
	
	\section{Covariance Between Estimates of Arbitrary Vertices} \label{sec:cov}

    This section provides an algorithm to compute $\cov(\tilde{\betab}(c),\tilde{\betab}(d))$ for arbitrary $c,d\in\tree.$ To do this, some additional notation will be helpful. First, let $ c \wedge d$ denote the unique vertex that is the closest common ancestor of vertices $c,d$ in $\tree.$ Also, let $\omega(c, d)$ denote the unique shortest path from $c$ to $d$ in $\tree.$ We also denote the $t^\textrm{th}$ element (respectively, the $t^\mathrm{th}$ element from the end) of this path by $\omega(c,d)[t-1]$ $(\omega(c,d)[-t]).$ In other words, for $c,d\in\tree,$ the shortest path from $c$ to $d$ is given by $\omega(c, d)[0]$ $(=c), \omega(c, d)[1],\dots, \omega(c,d)[-2], \omega(c, d)[-1]$ $(=d).$ Detailed derivations of $\cov(\tilde{\betab}(c),\tilde{\betab}(d))$ are provided in Appendix \ref{sec:appendix:cov_mat_derivations}, and the resulting covariances can be found in Algorithm \ref{alg:cov_est}.

	\begin{algorithm}[tbh]
		\DontPrintSemicolon
		\SetKwInOut{Output}{return}
		\SetKwInOut{Input}{input}
		\Input{$\tree:$ The rooted tree}
		\Input{$\{\tilde{\betab}(v), \var(\tilde{\betab}(v)), A(v),\var(\hat{\betab}(v  |  v-))\}_{v\in\tree}:$ The output of Algorithm \ref{alg:two_pass_est}}
		\Input{$c,d\in\tree:$ This algorithm returns $\cov(\tilde{\betab}(c), \tilde{\betab}(d))$}
		\If{$c=d$}{
			\Output{$\var(\tilde{\betab}(c))$}
		}
		\If{$d = c \wedge d $}{
			\tcp{See \eqref{eq:cov_ancestors} for this case:}
			\Output{$\left(\prod_{k\in\omega(c,d) / d} A(k) \right)\var(\tilde{\betab}(v))$} 
		}
		\If{$c=c\wedge d$}{
			\tcp{Switch the inputs $c,d$ and transpose the output of this function:}
			\Output{ComputeCovariance$ \left(\tree, \{\tilde{\betab}(v), \var(\tilde{\betab}(v)), A(v),\var(\hat{\betab}(v  |  v-))\}_{v\in\tree}, d,c \right)^\top$}
		}
		\tcp{The remaining case follows from \eqref{eq:cov_not_ancestors}:}
		$c'\gets \omega(c,c\wedge d)[-2]$ \;
		$d'\gets \omega(c \wedge d, d)[1]$ \; 
		$\cov(\tilde{\betab}(c'), \tilde{\betab}(d'))  \gets A(c') \var(\tilde{\betab}(c' \wedge d') )  A(d')^\top - A(c') \var(\hat{\betab}(d'| d'-))$ \;
		\Output{$ \left(\prod_{k \in \omega(c, c') / \{c'\}} A(k)\right) \cov(\tilde{\betab}(c'), \tilde{\betab}(d')) \left(\prod_{k \in \omega(d', d) / \{d'\}} A(k)^\top \right)$}
		\caption{ComputeCovariance: Computes the matrix $\cov(\tilde{\betab}(c), \tilde{\betab}(d))$ for arbitrary $c,d\in\tree$ } \label{alg:cov_est}
	\end{algorithm}
	
	The derivations in this section can be used to derive the off-diagonal blocks of $\var(\tilde{\betab}),$ as defined in \eqref{eq:var_betab_naive}, so, combined with the derivation for the diagonal blocks provided in the previous section, these results allow for the computation of arbitrary elements of the variance matrix of the standard GLS estimate. The next section describes how Algorithms \ref{alg:two_pass_est} and \ref{alg:cov_est} can be used to compute CIs for linear combinations of the unknown population parameter, $\betab.$ 

	\section{Confidence Interval Estimation} \label{sec:summary}

	The previous two sections describe each of the state variable updates required for the two-pass estimation approach and the computation of $\cov(\tilde{\betab}(c), \tilde{\betab}(d))$ for arbitrary $c,d\in\tree,$ as summarized in Algorithms \ref{alg:two_pass_est} and \ref{alg:cov_est}, respectively. Algorithm \ref{alg:est_ci} uses the output of these initial algorithms to produce a CI of the population parameter
	\begin{align}
		\beta_{H,\qb} = \sum_{v\in H} \qb^\top \betab(v),
	\end{align}
	where $\qb\in\rr^n$ and $H\subset \level(\tree, L)$ are user-defined inputs. Specifically, these CIs are centered at 
	\begin{align}
		\tilde{\beta}_{H,\qb} = \sum_{v\in H} \qb^\top \tilde{\betab}(v),
	\end{align}
	and the CI width is defined to ensure statistical validity. 

	After describing our implementation of Algorithm \ref{alg:est_ci} in the next subsection, we will provide a numerical experiment to assess the accuracy of these CIs in practice for our motivating use case. While Theorem \ref{thm:appendix:beta_tilde_is_normal} is not applicable in this use case because $\ub$ is not normally distributed, there are still reasons to be optimistic about the accuracy of the CIs output from Algorithm \ref{alg:est_ci}. First, as described in Section \ref{sec:setting}, the distribution of $\ub$ is closely related to the Gaussian distribution, \textit{i.e.}, it is the discrete Gaussian distribution \citep{canonne2020discrete}. Second, as described briefly in Appendix \ref{sec:appendix:beta_tilde_is_normal}, even when $\ub$ is not normally distributed, the resulting CIs are still statistically valid asymptotically as the number of observations, \textit{i.e.}, the dimension of $\yb$ in \eqref{eq:yb_def_intro}, diverges.

	\begin{algorithm}[tbh]
		\DontPrintSemicolon
		\SetKwInOut{Output}{return}
		\SetKwInOut{Input}{input}
		\Input{$\tree:$ The rooted tree }
		\Input{$\{\tilde{\betab}(v), \var(\tilde{\betab}(v)), A(v),\var(\hat{\betab}(v  |  v-))\}_{v\in\tree}:$ the output of Algorithm \ref{alg:two_pass_est}}
		\Input{$\qb\in\rr^n:$ The linear counting query of interest}
		\Input{$H\subset \level(\tree,L):$ The subset of the leaves of $\tree$ for which to estimate the linear query  }
		\Input{$\alpha:$ This algorithm will return the $1-\alpha$ CI of $\beta_{H,\qb}$}
        $J\gets H$
		\tcp{To reduce the runtime when $J$ contains many vertices, it can optionally be redefined:}
        \For{$l\in \{L-1, \dots, 0\}$}{
            \For{$v \in \level(\tree, l)$}{
                \If{$  \forall c \in \children(v), \; c\in J$}{
                    $J\gets \left( J / \children(v)\right) \cup \{v\}$ \;
                }
            }
        }
		$\var(\tilde{\betab}_{J,\qb})\gets 0 $\;
		$\tilde{\betab}_{J,\qb} \gets 0$\;
		\For{$c\in J$}{
			$\tilde{\betab}_{J,\qb} \gets \tilde{\betab}_{J,\qb} + \qb^\top \tilde{\betab}(c)$\;
			\For{$d\in J$}{
				$\cov(\tilde{\betab}(c),\tilde{\betab}(d))\gets$ ComputeCovariance$ \left(\tree, \{\tilde{\betab}(v), \var(\tilde{\betab}(v)), A(v),\var(\hat{\betab}(v  |  v-))\}_{v\in\tree}, c,d \right)$ \;
				$\var(\tilde{\betab}_{J,\qb})  \gets \var(\tilde{\betab}_{J,\qb}) + \qb^\top\cov(\tilde{\betab}(c),\tilde{\betab}(d)) \qb$\;
			}
		}
		\tcp{$\Phi(\cdot)$ denotes the cumulative distribution function of the standard normal distribution:}
		$k \gets  \sqrt{\var(\tilde{\betab}_{J,\qb})} \Phi^{-1}(1-\alpha/2) $ \;
		\Output{$\tilde{\betab}_{J,\qb}, (\tilde{\betab}_{J,\qb} - k, \tilde{\betab}_{J,\qb} + k)$}
		\caption{EstimateConfidenceInterval: Returns the estimate and CI of a user-defined query evaluated in an arbitrary subset of the leaf vertices of $\tree$} \label{alg:est_ci}
	\end{algorithm}

	\subsection{Computational Considerations} \label{sec:computational_considerations}
	
	For many use cases, the limiting factor for the feasibility of our proposed algorithms will be the need to compute inverses of $n\times n$ matrices. For example, one possible use case would be to apply these methods to the noisy measurements of the Redistricting Data File, and, since $n$ is equal to $2,016$ in this case, the algorithm can be applied directly because it is straightforward to invert $2,016 \times 2,016$ matrices. Since there are over 6 million vertices in the tree for our primary use case and it is ideal to store the output of Algorithm \ref{alg:two_pass_est} to avoid rerunning it in the future, our implementation of the methods described here also uses a few techniques to reduce the size of the matrices that are output from Algorithm \ref{alg:two_pass_est}, including using an output format that only stores the upper triangle of symmetric matrices and applying a compression algorithm before saving these outputs. These techniques also reduce the amount of data that must be transferred between nodes in our clusters, which improves runtime and stability. 

    Unlike the noisy measurements for Redistricting Data File DAS executions, the schemas used to generate the noisy measurements for the persons and units universes of the DHC DAS executions result in the dimension of $\betab(v)$ for each $v\in\tree$ being more than 1.2 million. Since it would not be possible to invert matrices with more than 1.2 million rows and columns, a preprocessing step must be carried out before executing Algorithm \ref{alg:two_pass_est} that marginalizes the noisy measurements to a lower dimension. As an example, the DHCP Table PCT12 provides Sex $\times $ Age tabulations for each geographic entity, using an Age recode that includes 103 Age categories. To estimate CIs for this table, one could marginalize the noisy measurements to this table schema to reduce the dimension of $\betab(v)$ to only 206. 

	Another algorithmic modification that is worth mentioning is one that reduces the computational cost of Algorithm \ref{alg:est_ci} when the geographic entity is far from the tree $\tree,$ \textit{i.e.}, when the cardinality of $J$ in Algorithm \ref{alg:est_ci} is high. Note that this algorithm's computational complexity scales quadratically in $\card(J)$ because $O(\card(J)^2)$ covariance matrices are computed in the nested for loop of this algorithm. Rather than performing the operations in this nested for loop, we can instead start at level $L$ of the tree and perform a fine-to-coarse recursion, with each step computing the set of estimates for entities, each with a geographic extent defined as the intersection of that of $\cup_{v\in H}v$ and that of $v\in \level(\tree, l).$ Specifically, for each level $l=\{L-1, \dots,0\},$ we can consider each parent vertex $ v\in\level(\tree, l)$ with children in $H,$ and redefine the attributes associated with this vertex that were defined in Algorithm \ref{alg:two_pass_est} as
	\begin{align*}
		\tilde{\betab}(v) \gets  \sum_{c\in\children(v)\cap H} \tilde{\betab}(c), \;& \;  \var(\tilde{\betab}(v)) \gets \var\left(\sum_{c\in\children(v)\cap H} \tilde{\betab}(c) \right), \\
		A(v) \gets \left(\sum_{c\in\children(v)\cap H} A(c)\right) A(v), \;& \textrm{and }  \var(\hat{\betab}(v)| v-) \gets  \var(\hat{\betab}(v)| v-) \left(\sum_{c\in\children(v)\cap H} A(c)\right)^\top.	
	\end{align*}
	Afterward, the elements of $\children(v)$ in $H$ are replaced with $v.$ It is straightforward to use the derivations in Appendix \ref{sec:appendix:cov_mat_derivations} to verify that this approach results in the same estimate and covariance as the one used in Algorithm \ref{alg:est_ci}, but it has a computational complexity that only scales quadratically in $\max_{v\in\tree} \card(\children(v)).$ Regardless, our implementation used the simpler approach described in Algorithm \ref{alg:est_ci} because our primary interest is in CIs for geographic entities on the tabulation spine, which results in $\card(J)$ being fairly low in all geographic entities of interest.

	\section{Numerical Experiment} \label{sec:numerical_exp}

	In this section we provide a numerical experiment using a Python/PySpark implementation of the algorithms we proposed above. Rather than using all noisy measurements generated by a full U.S. DAS execution, in this numerical experiment, we used the publicly available noisy measurements from the persons universe Redistricting Data File Puerto Rico (PR) DAS execution of the 2010 Demonstration Data Product; instructions for downloading these noisy measurements are available at \cite{DHCddps}. This DAS execution used settings that were identical to that of the 2020 production DAS execution that was used to publish the 2020 PR persons tables in the Redistricting Data File, other than the fact that the 2010 Census data was used as input instead of the 2020 Census data.

	This numerical experiment focuses on estimating CIs for 2010 census block groups. As described previously, the spine used internally in DAS executions (\textit{i.e.}, the rooted tree, $\tree)$ is distinct from the tabulation spine, and census block groups are the set of tabulation spine geographic units that are furthest from the internal DAS spine. As a result, census block groups provide an interesting set of off-spine geographic entities for the purpose of testing our proposed estimation approaches.

	We will consider CIs at the 0.9 and 0.95 confidence levels in particular. In the context of this numerical experiment, we will use \textit{empirical coverage}, or \textit{coverage}, of a set of CIs to refer to the proportion of the CIs in this set that contain the true population count. We will estimate the coverage of CIs for several sets of queries in this section. 

	Since we know each true count is nonnegative, we can optionally increase the coverage of each CI and reduce its width, by rounding the endpoints of the CI up to zero in the event either endpoint is negative, which we will refer to as a \textit{nonnegative CI}. To see that this improves the empirical coverage, consider an initial CI with both endpoints that are strictly negative. In these cases, since the confidential CEF-based count is always non-negative, this initial CI cannot contain the CEF-based count, but after rounding both endpoints up to zero, it will contain the CEF-based count when this count is equal to zero. Tables \ref{tab:coverage} and \ref{tab:width} provide the empirical coverage and the average width, respectively, of the 0.9 and 0.95 CIs, both before and after rounding up to zero, for three sets of queries. These tables use ``All Redistricting Tabulations" to refer to all persons Redistricting Data File tabulations published in 2020, \textit{i.e.}, tables P1, P2, P3, P4, and P5, as described by \cite{redDataFile}, ``Total Population" to refer to the total population, and ``CENRACE" to refer to the collection of queries that provide the population counts in each of the 63 race combination categories of this attribute, as described in Section \ref{sec:setting}. For each of these three sets of queries, Tables \ref{tab:coverage} and \ref{tab:width} consider the CIs for each query in the set, evaluated in each 2010 vintage PR tabulation census block group. 

	\begin{table}[tbh]
		\begin{tabular}{lllll} 
			& \multicolumn{4}{c}{CI Type}                          \\ \cline{2-5} 
			Query & \multicolumn{1}{c}{0.90 CI} & 0.95 CI & 0.90 CI, Nonnegative & 0.95 CI, Nonnegative \\ \hline
			All Redistricting Tabulations    & 0.8994 &	0.9496 &	0.9473 & 0.9737  \\
			Total Population                 & 0.8922 & 0.9363 & 0.8922 & 0.9363  \\
			CENRACE                          & 0.9002 & 0.9498 & 0.9479 & 0.9739
		\end{tabular}
        \caption{CI Empirical Coverage: For each of the four CIs considered, and each of the three sets of queries considered, this table provides the CI coverage, which is the proportion of CIs that contain the true CEF-based count out of all queries in the given query set and all block groups in PR.} \label{tab:coverage}
	\end{table}

	Table \ref{tab:coverage} shows that coverage of the CIs before rounding up to zero, has a coverage that is fairly close to the confidence level of the CI considered, for each of the three sets of queries. The final two columns demonstrate that rounding these endpoints up to zero improves the coverage of these CIs, particularly for sets of tabulations with confidential CEF-based counts that are highly sparse at the block-group geographic level, \textit{i.e.}, both the CENRACE tabulations and the union of all Redistricting Data File tabulations. Table \ref{tab:width} shows that rounding up to zero can also significantly narrow the widths of the CIs for these more granular queries. These CIs are also shown to be fairly wide on average. In part, this is because census block groups were not on the optimized spine used internally in Redistricting Data File DAS executions.\footnote{Several improvements to our proposed methods are described in Section \ref{sec:upcoming_ci_release} that significantly reduce the CI widths, including using the publicly-available linear equality constraints used within DAS and also leveraging the noisy measurements from the 2020 production DHC DAS execution. The CIs in the data product accompanying this paper are much narrower than the ones presented here in part because they were computed in a TwoPassGLS execution that leveraged these improvements.}

    \begin{table}[tbh]
		\begin{tabular}{lllll}
			& \multicolumn{4}{c}{CI Type}                          \\ \cline{2-5} 
			Query & \multicolumn{1}{c}{0.90 CI} & 0.95 CI & 0.90 CI, Nonnegative & 0.95 CI, Nonnegative \\ \hline
			All Redistricting Tabulations    & 100.8 & 120.1 & 54.84 & 65.74  \\
			Total Population                 & 314.4 & 374.6 & 313.1 & 373.2  \\
			CENRACE                          & 97.81 & 116.6 & 50.66 & 60.78
		\end{tabular}
        \caption{CI Averaged Widths: For each of the four CIs considered, and each of the three sets of queries considered, this table provides the CI width, averaged over all queries in the given query set and all block groups in PR.} \label{tab:width}
	\end{table}

	Note that the analyses above assess statistical validity of these CIs at only two confidence levels. To provide a way to assess statistical validity at all confidence levels, we also constructed the $Z-$scores of a set of estimates, which for each estimate $\tilde{\beta}_{H,\qb}$ is defined as,
	
	\begin{align}
		z_{H,\qb} = \left(\tilde{\beta}_{H,\qb} - \beta_{H,\qb} \right) / \sqrt{\var(\tilde{\beta}_{H,\qb})}.
	\end{align}
	
	\noindent
	If $\ub(v)$ were normally distributed for each $v\in\tree,$ then we would have $z_{H,\qb}\sim N(0,1),$ so we can asses the error due to our normality modeling assumption by computing the $Z-$scores for a set of queries to see how closely these scores appear to follow a standard normal distribution. To do this, we computed the $Z-$scores for the total population query for each of the 2,543 census block groups in PR. Figure \ref{fig:q_q_plot} provides a Q-Q plot based on these $Z-$scores. For each blue point, the horizontal axis provides the empirical quantiles of the 2,543 sample $Z-$scores and the vertical axis provides the corresponding theoretical quantile value based on a standard normal distribution. The red line provides the identity function. Since the blue points are very close to the identity function, this set of $Z-$scores appear to follow a Gaussian distribution quite closely. This plot also implies that, for any $\alpha\in (0,1),$ the $1-\alpha$ CIs for this set of queries would have empirical coverage close to $1-\alpha.$ 
	
	\begin{figure}[tbh]
		\centering
		\includegraphics[width=90mm]{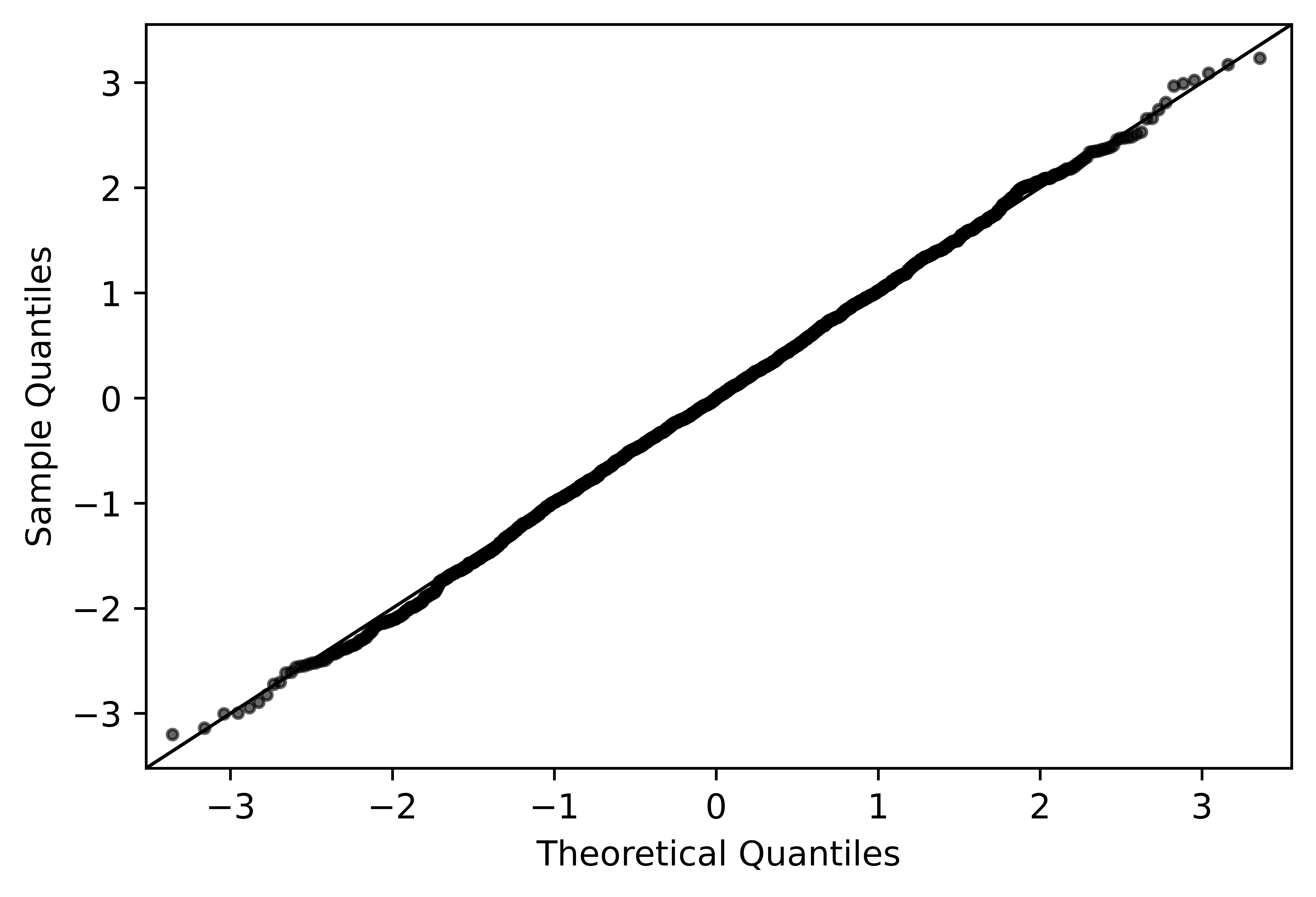}
		\caption{Q-Q Plot: The points in this plot can be viewed as points on a parametric curve, parameterized by the quantile $q\in(0,1).$ The horizontal axis provides the $q$ quantile of the standard normal distribution and the vertical axis provides the sample quantile of the $Z-$scores of the total population estimates for tabulation block groups. The line provides the identity function, $f(x) = x.$}
        \label{fig:q_q_plot}
	\end{figure}
	
	\section{Comparison with \texorpdfstring{\cite{hay2010boosting}}{Hay et al. (2010)}}
	\label{sec:connection_w_hays_et_al}
	
	As described above, the approach proposed in this paper can be viewed as a generalization of the approach proposed by \cite{hay2010boosting}, which we will describe in more detail in this section. To do so, we will assume we are in the setting described by \cite{hay2010boosting}. In other words, we will suppose that, for every $v\in\tree,$  $S(v)$ is equal to the $1\times 1 $ matrix, $[1],$ each vertex $v\in\tree$ has $k\in\nn$ children, and that the privacy loss budget allocated to each query, and therefore also the variance of each observation, is the same. Although it is not assumed by \cite{hay2010boosting}, we will also suppose the variance of each query is equal to one, since this variance does not appear in their algorithm, and this convention provides a slight simplification to the notation below.
	
	In this setting \cite{hay2010boosting} show that the GLS estimator can be found using the two-step approach reproduced in the following lemma using our notation.
	
	\begin{lemma} \label{lem:hay_boosting}
		(Theorem 3 from \cite{hay2010boosting}) If, for every $v\in\tree,$ $S(v)$ is equal to the matrix $[1],$ each vertex $v\in\tree$ has $k\in\nn$ children, and the variance of each observation is one, then the GLS solution can be found using the recurrence relation
		\begin{align} \label{eq:hay_boosting}
			\check{\betab}(v) =& \begin{cases} 
				z(v) & v\in\level(\tree, 0) \\
				z(v) + \frac{1}{k} \left( \check{\betab}(v) - \sum_{c\in\children(v)} z(c) \right) & v\not\in\level(\tree,0) \textrm{ and } v \in \children(v),
			\end{cases}
		\end{align}
		where
		\begin{align}
			z(v) = \begin{cases} 
				y(v) & v\in\level(\tree, L) \\
				\left(\frac{k^{L - l +1}- k^{L-l}}{k^{L-l+1}-1}\right) y(v) +  \left(\frac{k^{L-l} - 1}{k^{L-l+1}-1} \right)\sum_{c\in\children(v)} z(c) & v\in\level(\tree,l) \textrm{ and } l \neq L
			\end{cases}
		\end{align}
		% j(L-1) -> 2
		% j(L-2) -> 3
		% j(l) -> L - l + 1
	\end{lemma}
	
	Below we show that our proposed estimator is a generalization of the one proposed by \cite{hay2010boosting}.
	
	\begin{theorem} \label{thm:hay_et_al_equality}
		Suppose Assumptions \ref{assump:spine_struct}-\ref{assump:u_iid} hold and that, for every $v\in\tree,$ $S(v)$ is equal to the matrix $\begin{bmatrix} 1 \end{bmatrix}$ and $\var(\ub(v))$ is equal to the same scalar value. Also, suppose that each parent vertex $v\in\tree$ has exactly $k\in\nn$ children. Then the count estimator described by \cite{hay2010boosting} is identical to the count estimator $\{\tilde{\betab}(v)\}_{v\in\tree}$ proposed in Section \ref{sec:two_pass_est}.
	\end{theorem}
	\begin{proof}
		First we will show that $z(v)=\hat{\betab}(v|v-)$ under the assumptions of the theorem. To do so, for $v\in\level(\tree, l),$ we will begin by deriving a convenient functional form for $w(l)=\var(\hat{\betab}(v|v-))$ under the conditions of the theorem. Using the assumptions of this theorem, for any $v\in\level(\tree,L),$ we have $w(L)=\var(\yb(v))=1.$ Using notation introduced in Section \ref{sec:two_pass_est}, we have $\var(\hat{\betab}(v|\children(v)-))=k w(l)$ for any $v\in\level(\tree, l-1),$ so the variance of the inverse-variance-weighted mean implies 
		\begin{align*}
			w(l-1)&=\frac{1}{\left( 1+1/(k w(l))\right)} = 1-\frac{1}{1+k w(l)}
		\end{align*}
		\noindent
		After solving this recursive equation with the boundary condition that $w(L)=1,$ we have,
		\begin{align*}
			w(l)=\frac{(k-1) k^{L}}{k^{L+1}-k^l} % \label{eq:vofl}
		\end{align*}
		\noindent
		For any $v\in\level(\tree,l),$ \eqref{eq:up_g_g_minus1} implies
		\begin{align*} 
			\hat{\betab}(v | v-) = w(l) \left( \yb(v) + \frac{1}{k w(l+1)} \sum_{c\in\children(v)} \hat{\betab}(c | c-) \right). % \label{eq:betab_g_cond_g_minus_hay_sec_unsimplified} 
		\end{align*}
		
		Note that, for any $v\in\level(\tree,L),$ we have $z(v)=\hat{\betab}(v|v-).$ Also, in the case of vertices in levels above level $L,$ the definition of $z(v)$ and $ \hat{\betab}(v|v-)$ follows a similar form, in the sense that, for any $v\in\level(\tree,l)$ with $l<L,$ both estimators are weighted means of $\yb(v)$ and the sum over the estimates of the children of $v.$ In the case of both $z(v)$ and $ \hat{\betab}(v|v-),$ the weights of these two terms sum to one, so we will establish that $z(v)=\hat{\betab}(v|v-)$ for all $v\in\level(\tree,l)$ with $l<L$ by showing that the ratio of these weights are the same for both of these estimators. In the case of $\hat{\betab}(v | v-),$ this ratio is given by $k w(l+1),$ and in the case of $z(v)$ this ratio is given by 
		\begin{align*}
			\frac{k^{L - l +1}- k^{L-l}}{k^{L-l} - 1}.
		\end{align*}
		\noindent
		Thus, to establish that $z(v)=\hat{\betab}(v|v-)$ for all $v\in\level(\tree,l)$ with $l<L,$ we need to show
		\begin{align*}
			\frac{k^{L - l +1}- k^{L-l}}{k^{L-l} - 1} = k w(l+1) & \iff \frac{k^{L - l +1}- k^{L-l}}{k^{L-l} - 1} = \frac{(k-1) k^{L+1}}{k^{L+1}-k^{l+1}} \\
			& \iff \frac{k^L (k-1)}{k^L-k^l}=\frac{k^L (k-1)}{k^L-k^l},
		\end{align*}
		\noindent
		which implies $z(v)=\hat{\betab}(v|v-)$ for all $v\in\tree.$ 
		
		The final result follows from the fact that, after replacing $z(v)$ with $ \hat{\betab}(v|v-)$ in the recurrence relation in \eqref{eq:hay_boosting}, this recursive formula is equal to the coarse-to-fine recursion described in Section \ref{sec:two_pass_est}.
	\end{proof}

	\section{Production 2020 CIs Data Product} \label{sec:upcoming_ci_release}
	
	To allow the public to benefit from the CI method proposed in this paper without implementing and running the algorithms, the DAS team executed our implementation of these methods on the 2020 production noisy measurements to estimate CIs for each query in the persons universe of the 2020 Redistricting Data File at the tract level and above in the tabulation U.S. and Puerto Rico spines. To provide a sense of the number of CIs in this data product, a total of 292 CIs were produced for each tabulation geographic unit, and the U.S. tabulation spine includes 83,883 tracts, 3,143 counties, 51 states, and the root geographic unit consisting of the U.S. as a whole.\footnote{The Census Bureau Geography Division includes Washington DC in the state geographic level as a ``state equivalent" geographic unit.} These CIs and directions for using them are available at \url{https://registry.opendata.aws/census-2020-pl94-gls/}. 
	
    We made several changes to our implementation prior to this execution to improve accuracy relative to the implementation used in the numerical experiment in the previous section. First, to leverage additional information encoded in the equality constraints used by DAS, which are also available in the publicly-available noisy measurements, for each vertex $v\in\tree$ that is above the block geographic level, we added additional rows to $S(v)$ and elements to $\yb(v)$ that encode the equality constraints associated with $v$ used within the DAS production executions. While we may explore implementing these constraints using a Karush–Kuhn–Tucker (KKT) matrix in the future, the version of our codebase used for this particular data product only ensures these constraints hold approximately, by defining the variance of these additional elements of $\yb(v)$ to be 1/16,384.

	Second, this implementation actually executes Algorithm \ref{alg:two_pass_est} two times, once using  the persons Redistricting Data File noisy measurements as input and once using the DHCP noisy measurements as input, marginalized to the persons redistricting data file schema consisting of 2,016 histogram cell counts for each vertex. Afterward, the final set of CIs for each geographic entity are computed by combining the output of these two executions in an optimal manner. Specifically, for each execution $i\in\{0,1\}$ of Algorithm \ref{alg:two_pass_est}, and for each geographic entity $H,$ the approach proposed in this paper is used to compute the histogram estimate $\tilde{\betab}_H^{(i)}\in\rr^n.$ Afterward, these two estimates are combined with the optimal inverse-variance-weighted mean
	
	\[\tilde{\betab}_H = \var(\tilde{\betab}_H) \left( \var\left(\tilde{\betab}_H^{(0)}\right)^{-1} \tilde{\betab}_H^{(0)} + \var\left(\tilde{\betab}_H^{(1)}\right)^{-1} \tilde{\betab}_H^{(1)} \right), \]
	where
	\[\var(\tilde{\betab}_H) = \left( \var\left(\tilde{\betab}_H^{(0)}\right)^{-1}  + \var\left(\tilde{\betab}_H^{(1)}\right)^{-1} \right)^{-1},\]
	as described by Lemma \ref{lem:inv_var_wtd_mean} in more detail, and then the final CI for each query $q$ is computed using a similar approach as in Algorithm \ref{alg:est_ci}, \textit{i.e.}, by $q^\top \tilde{\beta}_{H}\pm  \Phi^{-1}(1-\alpha/2) \sqrt{q^\top \var(\tilde{\beta}_{H}) q}.$ 

	To provide an idea of the computational requirements of our implementation of the proposed algorithms for the numerical experiment explored here, we will describe the cluster settings and runtimes required for our implementation of this algorithm. The codebase itself is available at \url{https://github.com/uscensusbureau/DAS_2020_GLS_Uncertainty_Evaluation}. This implementation was executed on Amazon Web Services (AWS) using Elastic Map Reduce (EMR) version 6.15 and Apache Spark version 3.4. Each node in the cluster was an AWS r6i.8xlarge virtual machine, which has 32 virtual cores and 244GiB of RAM. We configured the cluster so that it had one primary node and 40 core nodes, and configured Spark so that there were 40 executors, each with 11 spark executor cores. Computing the CIs required first executing Algorithm \ref{alg:two_pass_est} using the 2020 Persons Redistricting Data File noisy measurements as input, which had a runtime of 9.7 hours. Next, we marginalized the 2020 DHCP noisy measurements to the schema used by the Persons Redistricting Data File and then used these marginalized noisy measurements as input in Algorithm \ref{alg:two_pass_est}, which had a runtime of 31.7 hours. In the final step we computed the histogram estimate and variance matrix of each tabulation geographic unit at the tract level and above separately in both of these two spines and combined them using the approach outlined in the previous paragraph, which had a runtime of 16.0 hours.

	\section{Conclusion} \label{sec:conclusion}
	
	This paper describes a two-pass estimation approach for hierarchical data that is capable of providing GLS-based estimates and CIs. We also provide a numerical exercise to demonstrate feasibility of our proposed methods in our motivating use case. As described in Section \ref{sec:connection_w_hays_et_al}, this two-pass estimator can be viewed as a generalization of the approach described by \cite{hay2010boosting}. We also describe a statistical data product based on our proposed approach in Section \ref{sec:upcoming_ci_release}. 
	
	In the context of our motivating use case, an alternative to using the GLS estimator described above is to simply derive point estimates from the noisy measurements directly. For example, one can estimate the total population of a vertex using the total population primitive DP answer of the vertex directly. An alternative estimator for this same query answer can be derived by summing over the total population noisy measurements of the children of this parent vertex. For use cases that do not require a high level of accuracy, these simple estimates may be adequate. However, the full-information GLS estimator provides a unique estimate without requiring further choices of the user, which has the effect of enhancing reproducibility. As described in \ref{sec:main_results}, the GLS estimate has the added benefit of being the best linear unbiased estimator, which results in CIs that are narrower than those of these simpler estimators.

	\bibliographystyle{apalike}
	\bibliography{references}
	
	\appendix
	\section{Derivation of Algorithm for Computing the Two-Pass GLS Estimator} \label{sec:derivation_of_two_pass_algorithm}

	\subsection{Fine-to-Coarse Recursion: Initialization Step} \label{sec:bottom_up_pass_init}
	
	The first step of the algorithm is to initialize state variables of each vertex $v\in\tree$ as
	
	\begin{align}
		\hat{\betab}(v | v) &= \left(S(v)^\top \var(\ub(v))^{-1} S(v)\right)^{-1} S(v)^\top \var(\ub(v))^{-1} \yb(v), \label{eq:up_init1} \\
		\var(\hat{\betab}(v | v)) &= \left(S(v)^\top \var(\ub(v))^{-1} S(v)\right)^{-1}. \label{eq:up_init2}
	\end{align}
	
	\noindent
	For each of the leaf of the tree, \textit{i.e.}, each vertex in $\level(\tree, L),$ also let 
	
	\begin{align}
		\hat{\betab}(v | v-) &= \hat{\betab}(v | v) \label{eq:up_init3} \\
		\var(\hat{\betab}(v | v-)) &= \var(\hat{\betab}(v | v)) \label{eq:up_init4}
	\end{align}
	
	\subsection{Fine-to-Coarse Recursion: Recursion Step} \label{sec:bottom_up_pass_recursion_step}
	
	The updates described in this section are performed first for vertices in $\level(\tree, L-1),$ and afterward in each level moving up the tree. Specifically, given $\{\hat{\betab}(c| c-)\}_{c\in\children(v)}$ and $\{\var(\hat{\betab}(c | c-))\}_{c\in\children(v)},$ let
	
	\begin{align}
		\hat{\betab}(v | \children(v)-)  &= \sum_{c\in\children(v)} \hat{\betab}(c| c-), \\
		\var(\hat{\betab}(v | \children(v)-))  &= \sum_{c\in\children(v)} \var(\hat{\betab}(c| c-)).
	\end{align}
	
	Now, since we have two independent estimates of $\betab(v),$ \textit{i.e.}, $\hat{\betab}(v | v)$ and $\hat{\betab}(v | \children(v)-),$ for each vertex $v,$ we can define a new estimate as the linear combination of these estimates that has the lowest possible variance, which is the inverse-variance weighted mean of these two vector estimates, as described in Lemma \ref{lem:inv_var_wtd_mean}. In other words, in this step, we update the state variables using
	
	\begin{align} 
		&\var(\hat{\betab}(v | v-)) = \left(\var(\hat{\betab}(v | v))^{-1} + \var(\hat{\betab}(v | \children(v)-))^{-1}\right)^{-1} \label{eq:up_g_g_minus1} \\
		&\hat{\betab}(v | v-)  = \var(\hat{\betab}(v | v-)) \left(\var(\hat{\betab}(v | v))^{-1} \hat{\betab}(v | v) + \var(\hat{\betab}(v | \children(v)-))^{-1} \hat{\betab}(v | \children(v)-) \right). \label{eq:up_g_g_minus2}
	\end{align}
	
	\subsection{Coarse-to-Fine Recursion: Initialization Step}
	
	Prior to starting the coarse-to-fine recursion, we initialize the state variables of the root vertex $r$ as
	
	\begin{align}
		\tilde{\betab}(r) &= \hat{\betab}(r |r-) \label{eq:down_init1}  \\
		\var(\tilde{\betab}(r)) &= \var(\hat{\betab}(r | r-)).\label{eq:down_init2}
	\end{align}
	
	\subsection{Coarse-to-Fine Recursion: Recursion Step} \label{sec:top_down_pass_recursion_step}
	
	The updates in the coarse-to-fine recursion can be viewed as solutions to the optimization problem
	
	\begin{align}
		\{\tilde{\betab}(c)\}_{c\in\children(v)} &= \argmin_{ \{\betab(c)\}_{c\in\children(v)} }  \sum_{c\in\children(v)}  \lVert \var(\hat{\betab}(c | c-))^{-1/2}  \left(\hat{\betab}(c | c-) - \betab(c)\right) \rVert_2^2 \textrm{  such that:} \nonumber \\
		& \sum_{c\in\children(v)} \betab(c) = \tilde{\betab}(v). \label{eq:top_down_proj_opt_prob}
	\end{align}
	
	The KKT conditions of this optimization problem are
	
	\begin{align*}
		& \var(\hat{\betab}(c | c-))^{-1}  \left(\tilde{\betab}(c) - \hat{\betab}(c | c-)\right) = \lambdab \;\; \forall \; c\in\children(v) \\
		& \sum_{c\in\children(v)} \tilde{\betab}(c) = \tilde{\betab}(v).
	\end{align*}
	
	\noindent
	These two conditions imply that $\lambdab$ can be found using 
	
	\begin{align*}
		& \sum_{c\in\children(v)} \tilde{\betab}(c)-\hat{\betab}(c | c-) = \sum_{c\in\children(v)} \var(\hat{\betab}(c | c-)) \lambdab \\
		\implies & \tilde{\betab}(v) - \sum_{c\in\children(v)} \hat{\betab}(c | c-) = \sum_{c\in\children(v)} \var(\hat{\betab}(c | c-)) \lambdab \\
		\implies & \lambdab = \left( \sum_{c\in\children(v)} \var(\hat{\betab}(c | c-)) \right)^{-1} \left( \tilde{\betab}(v) - \sum_{c\in\children(v)} \hat{\betab}(c | c-) \right).
	\end{align*}
	
	\noindent
	After substituting this value of $\lambdab$ into the first KKT condition, we have
	
	\begin{align}
		\tilde{\betab}(c) = \hat{\betab}(c | c-) + A(c) \left(\tilde{\betab}(v) - \sum_{c'\in\children(v)} \hat{\betab}(c' | c'-)\right), \label{eq:down1}
	\end{align}
	
	\noindent
	where
	\begin{align}
		A(c) = \var(\hat{\betab}(c | c-)) \left( \sum_{c'\in\children(v)} \var(\hat{\betab}(c' | c'-)) \right)^{-1}. \label{eq:A_of_c_def}
	\end{align}
	
	Next we will derive $\var(\tilde{\betab}(c)).$ To do so, first let $\widetilde{Q}_v$ be defined so that $\tilde{\betab}(v) =  \tau + \widetilde{Q}_v \sum_{c\in\children(v)} \hat{\betab}(c | c-) ,$ where $\tau$ is a random variable that is independent of $\{\hat{\betab}(c | c-)\}_{c\in\children(v)},$ \textit{i.e.}, $\tau$ is the component of $\tilde{\betab}(v)$ that is a linear function of the observations of vertex $v$ and the vertices that are ancestors of $v.$ Also, let $ B_c=\hat{\betab}(c | c-) - A(c) \sum_{c'\in\children(v)} \hat{\betab}(c' | c'-),$ so that $\tilde{\betab}(c) = B_c + A(c)\tilde{\betab}(v)$. To derive $\var(\tilde{\betab}(c)),$ we will first derive a two intermediate properties. First, the definition of  $A(c) $ implies
	
	\begin{align} \label{eq:ac_def_implication}
		A(c) \left( \sum_{c'\in\children(v)} \var( \hat{\betab}(c' | c'-)) \right) = \var( \hat{\betab}(c | c-)).
	\end{align}
	
	\noindent
	Second, we have
	
	\begin{align}
		&\cov(B_c, \tilde{\betab}(v)) = \cov\left(\hat{\betab}(c | c-) - A(c) \sum_{c'\in\children(v)} \hat{\betab}(c' | c'-) ,  \tau   + \widetilde{Q}_v \sum_{c'\in\children(v)} \hat{\betab}(c' | c'-) \right)\nonumber \\
		& = \var(\hat{\betab}(c | c-)) \widetilde{Q}_v^\top - A(c) \left(\sum_{c'\in\children(v)} \var(\hat{\betab}(c' | c'-) )\right) \widetilde{Q}_v^\top \nonumber \\
		& =  \left( \var(\hat{\betab}(c | c-)) - \var(\hat{\betab}(c | c-)) \right) \widetilde{Q}_v^\top  = 0, \label{eq:cov_b_and_tilde_beta_is_zero}
	\end{align}
	
	\noindent 
	where the penultimate equality above follows from \eqref{eq:ac_def_implication}.  

    Next, \eqref{eq:cov_b_and_tilde_beta_is_zero} implies that $\var(\tilde{\betab}(c))$ can be written as
	
	\begin{align}
		\var(\tilde{\betab}(c)) =& \cov(\tilde{\betab}(c), \tilde{\betab}(c)) = \cov(B_c + A(c) \tilde{\betab}(v), B_c + A(c) \tilde{\betab}(v) ) \nonumber \\
		=& \cov(B_c, B_c )+ \cov(B_c, A(c) \tilde{\betab}(v) ) + \cov(A(c) \tilde{\betab}(v), B_c) + \cov( A(c) \tilde{\betab}(v),  A(c) \tilde{\betab}(v) ) \nonumber \\
		=& \var(B_c) + A(c) \var(\tilde{\betab}(v)) A(c)^\top \nonumber  \\
		=& \var( \hat{\betab}(c | c-)) -  \var( \hat{\betab}(c | c-)) A(c)^\top   - A(c) \var( \hat{\betab}(c | c-)) \nonumber \\
		&+ A(c) \left(\sum_{c'\in\children(v)} \var(\hat{\betab}(c' | c'-) )\right) A(c)^\top + A(c)\var( \tilde{\betab}(v)) A(c)^\top, \nonumber 
	\end{align}
	
	\noindent
	so \eqref{eq:ac_def_implication} implies
	
	\begin{align}
		\var(\tilde{\betab}(c)) =& \var( \hat{\betab}(c | c-)) -  \var( \hat{\betab}(c | c-)) A(c)^\top   - A(c) \var( \hat{\betab}(c | c-)) \nonumber \\
		& + \var( \hat{\betab}(c | c-)) A(c)^\top + A(c)\var( \tilde{\betab}(v)) A(c)^\top \nonumber \\ 
		=&  \var( \hat{\betab}(c | c-))   - A(c) \var( \hat{\betab}(c | c-))  + A(c)\var( \tilde{\betab}(v)) A(c)^\top. \label{eq:down2}
	\end{align}
	This final equality is provides the final definition of the outputs of Algorithm \ref{alg:two_pass_est}. The next section provides a proof that the estimator output from this algorithm is identical the GLS estimator defined in \eqref{eq:est_betab_naive}.

    \section{Covariance Matrix Derivations} \label{sec:appendix:cov_mat_derivations}

    In this section, we will derive the covariance matrix $\cov(\tilde{\betab}(c), \tilde{\betab}(d))$ for the three types of adjacency relationships between vertices $c,d\in\tree$ described in Figure \ref{fig:adj_struct_cov}. To do so, we will use the notation outlined in Section \ref{sec:cov} and Appendix \ref{sec:derivation_of_two_pass_algorithm}. It may be helpful to point out that the derivations for each covariance matrix $\cov(\tilde{\betab}(c),\tilde{\betab}(d))$ for $c,d\in\tree$ in this section follow a similar pattern as the derivations of $\var(\tilde{\betab}(c))$ in the preceding section. Specifically, each such derivation involves first recursively substituting instances of $\tilde{\betab}(\cdot)$ in the covariance function inputs for its definition in \eqref{eq:down1}, until the only instance of $\tilde{\betab}(\cdot)$ is the closest common ancestor of $c$ and $d$ in the tree $\tree.$ While initially this adds a significant number of terms to the expression of $\cov(\tilde{\betab}(c),\tilde{\betab}(d)),$ it has the advantage of simplifying the dependency relationships between all possible pairs of random variables that appear in the resulting formula. The second step is to simplify each covariance matrix expression using \eqref{eq:ac_def_implication}, as is also done in the previous section to derive $\var(\tilde{\betab}(c)).$ 

    \begin{figure}[tbh]
		\centering
		\begin{subfigure}{0.215\textwidth}
			\begin{tikzpicture}
				[thick,level distance=20mm, sibling distance=30mm]
				\node (uv) {$c\land d$}
				child { node (u) {$c$} } 
				child { node (v) {$d$} };
			\end{tikzpicture}
			\caption{ }
		\end{subfigure}%
		~ 
		\begin{subfigure}{0.105\textwidth} 
			\begin{tikzpicture}
				[thick,level distance=20mm, sibling distance=30mm]
				\node (uv) {$d = c\land d$}
				child { node (u) {$c$}  edge from parent[dotted, very thick] } ;
			\end{tikzpicture}
			\caption{ }
		\end{subfigure}%
		~ 
		\begin{subfigure}{0.215\textwidth}
			\begin{tikzpicture}
				[thick,level distance=10mm, sibling distance=15mm]
				\node (uv) {$c\land d $}
				child {node (up) {$c'$}  child { node (u) {$c$}  edge from parent[dotted, very thick] } child[transparent] { node (ue) { } } }
				child {node (vp) {$d'$} child[transparent] { node (ve) { } } child { node (v) {$d$}   edge from parent[dotted, very thick] }  };
			\end{tikzpicture}
			\caption{ } \label{fig:adj_struct_cov:c}
		\end{subfigure}
		\caption{In this section we derive $ \cov(\tilde{\betab}(c), \tilde{\betab}(d))$ for the three types of adjacency relationships between $c,d\in\tree$ represented in the subplots above: $c$ and $d$ are siblings (A); $d$ is an ancestor vertex of $c$ (B); and cases in which neither $c$ or $d$ are a direct descendant of the other (C). Dotted edges are used to denote the portions of $\omega(c,d)$ that include an arbitrary number of vertices.} \label{fig:adj_struct_cov}
	\end{figure}

    First, consider the case in which $c,d\in\tree$ are sibling vertices, \textit{i.e.}, $c,d\in\children(c\land d),$ and let $ B_v=\hat{\betab}(v | v-) - A(v) \sum_{v'\in\children(c\land d)} \hat{\betab}(v' | v'-)$ for each $v\in\{c,d\}.$ In this case, property \eqref{eq:ac_def_implication} implies
	\begin{align}
		\cov(B_c, B_d) =&  A(c)\left(\sum_{v\in\children(c \wedge d)} \var(\hat{\betab}(v | v-))\right) A(d)^\top \nonumber \\
		& -\var(\hat{\betab}(c| c-)) A(d)^\top - A(c) \var(\hat{\betab}(d| d-)) \nonumber \\
		=&  \var(\hat{\betab}(c| c-)) A(d)^\top -\var(\hat{\betab}(c| c-)) A(d)^\top - A(c) \var(\hat{\betab}(d| d-)) \nonumber \\
		=& - A(c) \var(\hat{\betab}(d| d-)), \label{eq:cov_sibling_intmdt_res}
	\end{align}
	so
	\begin{align}
		\cov(\tilde{\betab}(c), \tilde{\betab}(d)) =& \cov(B_c + A(c) \tilde{\betab}(c \wedge d), B_d + A(d) \tilde{\betab}(c \wedge d) ) \nonumber \\
		=& \cov(B_c, B_d )+ \cov(B_c, A(d) \tilde{\betab}(c \wedge d) )  \nonumber \\
		& + \cov(A(c) \tilde{\betab}(c \wedge d), B_d)  + \cov( A(c) \tilde{\betab}(c \wedge d),  A(d) \tilde{\betab}(c \wedge d) ) \nonumber \\
		=& A(c) \var(\tilde{\betab}(c \wedge d) )  A(d)^\top - A(c) \var(\hat{\betab}(d| d-)). \label{eq:cov_sibling}
	\end{align}

	Second, to derive $ \cov(\tilde{\betab}(c), \tilde{\betab}(d))$ for the case in which $d$ is an ancestor of $c,$ first suppose we are in the subcase in which $c\in\children(d).$ In this subcase we have 
	\begin{align}
		\cov(\tilde{\betab}(c), \tilde{\betab}(d)) &= \cov(B_c + A(c) \tilde{\betab}(d), \tilde{\betab}(d) ) = \cov(B_c, \tilde{\betab}(d) ) + A(c)\cov(  \tilde{\betab}(d), \tilde{\betab}(d) )  \nonumber \\
		& = A(c)\var ( \tilde{\betab}(d) ). \label{eq:cov_gc}
	\end{align}
	The covariance for the case in which $d=c\land d$ but $c\notin \children(d)$ can be derived in a similar manner.\footnote{Note that the formulas for $\cov(\tilde{\betab}(d), \tilde{\betab}(c))$ for cases in which $d$ an ancestor of $c$ also follow from our derivations here because $ \cov(\tilde{\betab}(d), \tilde{\betab}(c))=\cov(\tilde{\betab}(c), \tilde{\betab}(d))^\top.$} Specifically, as in the derivation for \eqref{eq:cov_gc}, for each $t\in\{\omega(c,d)[0], \dots,\omega(c,d)[-2]\},$ the term $B_{t}$ in the first input of the covariance function can simply be ignored because these terms are independent of the second input to $\cov(\cdot)$ (\textit{i.e.}, $\tilde{\betab}(c\land d)=\tilde{\betab}(d))$ by definition of each $B_t,$ which leads to the simple formula
	\begin{align}
		\cov(\tilde{\betab}(c), \tilde{\betab}(d)) &= \cov\left(\left( \prod_{k\in\omega(c,d) / d} A(k)\right) \tilde{\betab}(d), \tilde{\betab}(d) \right)   \nonumber \\
		&=\left(\prod_{k\in\omega(c,d) / d} A(k) \right)\var(\tilde{\betab}(d)). \label{eq:cov_ancestors}
	\end{align}

	For the third case we consider, suppose neither $c$ or $d$ are a direct descendant of the other, which is a generalization of the first case we considered. Let $c'=\omega(c,c\land d)[-2]$ and $d'= \omega(c\land d, d)[1],$ as depicted in Figure \ref{fig:adj_struct_cov:c}. To derive $\cov(\tilde{\betab}(c), \tilde{\betab}(d))$ for all $c,d$ vertex pairs included in this case, note that \eqref{eq:cov_ancestors} implies
	\begin{align*}
		\cov(\tilde{\betab}(c), \tilde{\betab}(c'))=\left(\prod_{k\in\omega(c,c') / c'} A(k) \right)\var(\tilde{\betab}(c')),
	\end{align*}
	and likewise
	\begin{align*}
		\cov(\tilde{\betab}(d'), \tilde{\betab}(d))=\var(\tilde{\betab}(d'))\left(\prod_{k\in\omega(d,d') / d'} A(k) \right)^\top.
	\end{align*}
	Since \eqref{eq:down1} is linear in $\tilde{\betab}(v)$ and the elements of $\{\tilde{\betab}(t)\}_{t\in \omega(c,c')/c'}$ are independent of those of $\{\tilde{\betab}(t)\}_{t\in \omega(d',d)/d'}$ after conditioning on $ \tilde{\betab}(c') $ and $\tilde{\betab}(d'),$ for this case, we have
	
	\begin{align}
		\cov(\tilde{\betab}(c), \tilde{\betab}(d)) &= \left(\prod_{k\in\omega(c,c') / c'} A(k) \right) \cov(\tilde{\betab}(c'), \tilde{\betab}(d')) \left(\prod_{k\in\omega(d,d') / d'} A(k) \right)^\top, \label{eq:cov_not_ancestors}
	\end{align}
	where $\cov(\tilde{\betab}(c'), \tilde{\betab}(d'))$ is defined as in \eqref{eq:cov_sibling}. 

	\section{Proof Two-Pass Estimator is the BLUE} \label{sec:blue_proof}

	\subsection{Additional Notation} \label{sec:appendix:notation}
	
	The proofs in this appendix consider GLS estimators based on a variety of different subsets of the observations, and more notation will be helpful to define these estimators. After providing these definitions, we summarize this notation in Table \ref{tab:appendix:notation} below. Note that throughout this appendix we assume Assumptions \ref{assump:spine_struct}-\ref{assump:u_iid} hold, but we do not require Assumption \ref{assump:normality}.
	
	First, we will define a few subsets of vertices of $\tree.$ These will primarily be used in the definitions of subsets of objects associated with the vertices of $\tree,$ such as to denote subsets of the observations, error terms, and the row vectors of the design matrices corresponding to these  observations. Specifically, for a given parent vertex $v\in\tree,$ let $\gamma$ be defined as the set of vertices that are descendants of $v,$ $\gamma_L$ be defined as the set of leaf vertices that are descendants of $v,$ $\gamma_C$ as shorthand for $\children(v).$ Also, let $\tau_L = \level(\tree, L).$ 
	
	As described previously, we use these sets of vertices to refer to subsets of objects. Specifically, given the set of vertices $W,$ let $\betab_{W}=\stack(\{\betab(w)\}_{w\in W}),$ $\ub_W= \stack(\{\ub(w)\}_{w\in W}),$ and $\yb_{W} =\stack(\{\yb(w)\}_{w\in W}).$ Also, let $V_{W}$ be defined as the block diagonal matrix $\var(\ub_{W}).$ These sets are also used to distinguish between several design matrices. Specifically, given the set of vertices $W,$ let $F_W$ be defined as a block matrix with block $(i,j)$ given by $S(W[i])$ if $\left(W\cap \tau_L\right)[j]$ is either equal to $W[i]$ or is a descendant of $\gamma[i],$ and a block of zeros otherwise. 
	
	Throughout this appendix, we will continue to assume that every set of vertices is totally ordered. As described in Section \ref{sec:setting}, this is primarily to ensure that, for a given set of vertices $W,$ the vertex order used to define $F_W,$ $\ub_W,$ and $\yb_W$ match one another. Also, solely to keep our notation slightly more concise in Appendix \ref{sec:appendix:tdp}, for each set of vertices $W$ that we consider, we assume that if $W$ has any ancestors or descendants of vertex $v,$ then $W$ is defined so that ancestors of $v$ and $v$ itself are ordered before all other vertices and that the descendants of $v$ are ordered after all other vertices. We refer to a vertex set that is ordered in this way as being \textit{ordered with respect to} $v,$ and we describe the implication of this assumption that simplifies our notation slightly at the beginning of Appendix \ref{sec:appendix:tdp}. 
	
	We will also use a few matrices to define elementwise sums over the objects associated with sets of vertices. Specifically, for vertex subsets $W$ and $X,$ let $J_{W,X}$ be defined as a block matrix with block $(i,j)$ given by $I_n$ if $W[j]$ either is a descendant of $X[i]$ or is equal to $X[i]$ itself and a block of zeros otherwise. For example, this definition implies $\jlc \betab_{\gamma_L} = \betab_{\gamma_C}$ and also $\jrc \betab_{\tau_L} = \betab_{\gamma_C}.$ Note also that this definition implies that if $W= \level(\tree, l),X= \level(\tree, l'), $ and $Y= \level(\tree, l''),$ where $L\geq l>l'>l''\geq 0,$ then $J_{X,Y}J_{W,X} = J_{W,Y}.$
	
	\begin{table}[ht]
		\centering
		\begin{tabular}{ l l }
			$\gamma$ & $\{v\in \tree\} / \{v\} $ \\
			$\gamma_L$ & $\level(\tree, L) \cap \gamma$ \\
			$\gamma_C$ & $\children(v)$ \\
			$\tau_L$  & $\level(\tree, L)$ \\
			$\betab_{W} $ & $\stack(\{\betab(w)\}_{w\in W})$ \\
			$\ub_W $ & $ \stack(\{\ub(w)\}_{w\in W}) $ \\ 
			$\yb_{W} $ & $  \stack(\{\yb(w)\}_{w\in W})$ \\
			$V_W$ & $\var(\ub_W) $\\
			$F_W$ & The design matrix encoding the linear queries associated with vertices in $W$ \\
			$J_{W,X}$ & Defined so that $J_{W,X}\betab_{W} = \betab_X$
		\end{tabular}
		\caption{Additional notation used in the Appendix \ref{sec:blue_proof}. Note that $W$ and $X$ are totally ordered sets of vertices in $\tree$.} \label{tab:appendix:notation}
	\end{table}
	
	The next example provides the matrix $\eg,$ the GLS estimator of $\betab_{\gamma_L}$ based on the observations associated with vertices of $\tree_v,$ which is denoted by $\hat{\betab}_{\gamma_L},$ and the matrices $\jlc$ and $\jlg$ for the small example sub-graph shown in Figure \ref{fig:appendix:eg_spine_intro}.
	
	Also, we use $A^+$ to denote the Moore-Penrose pseudoinverse of the matrix $A\in\rr^{M\times N}.$ Given a positive semidefinite matrix $A\in\rr^{M\times N},$ let $A^{1/2}$ denote a matrix that satisfies $A = A^{1/2} (A^{1/2})^\top.$ Also, given two sets $A,B$ we use $A/B$ to denote the relative complement of $A$ in $B.$
	
	Since we make several references to Aitken's Theorem \citep{aitken1935iv} in the proofs in this Appendix, the next Lemma provides this result for completeness. Solely for notational convenience, we provide this theorem using the notation for the full-information model $\yb_{\tree} = \er \betab_{\tau_L} + \ub_\tree$ as an example. This result is closely related to the Gauss-Markov Theorem; for example, in this particular model, Aitken's Theorem simply follows from applying the Gauss-Markov Theorem to the transformed model $V_{\tree}^{-1/2} \yb_{\tree} = V_{\tree}^{-1/2}\er \betab_{\tau_L} + V_{\tree}^{-1/2}\ub_{\tree},$ which has homoscedastic errors, as required by the Gauss-Markov Theorem, since $\var(V_{\tree}^{-1/2}\ub_{\tree}) = V_{\tree}^{-1/2} V_{\tree} (V_{\tree}^{-1/2})^\top = I.$
	
	\begin{lemma} \label{lem:appendix:gauss_markov_thm}
		(Aitken's Theorem) \citep{aitken1935iv}. Using the notation above for the model $\yb_{\tree} = \er \betab_{\tau_L} + \ub_{\tree}$ for convenience, the GLS estimator
		\begin{align}
			\hat{\betab}_{\tau_L} = \argmin_{\betab\in\rr^{n \card(\tau_L)}} \lVert  V_{\tree}^{-1/2} (\er \betab - \yb_{\tree}) \rVert_2^2  = (\er^\top V_{\tree}^{-1} \er)^{-1}\er^\top V_{\tree}^{-1} \yb_{\tree} \label{eq:appendix:aitkens_thm1}
		\end{align}
		\noindent
		is the BLUE of $\betab_{\tau_L}.$ In other words, if $\check{\betab}_{\tau_L}$ is an alternative estimator that is also linear in $\yb_{\tree},$ \textit{i.e.}, there exists  a matrix $C$ such that $\check{\betab}_{\tau_L} =C\yb_{\tree},$  and is also unbiased, \textit{i.e.}, $E(\check{\betab}_{\tau_L})=\betab_{\tau_L},$ then $\var(\hat{\betab}_{\tau_L}) = (\er^\top V_{\tree}^{-1} \er)^{-1}$ satisfies
		\begin{align}
			& \var(\hat{\betab}_{\tau_L}) \leq \var(\check{\betab}_{\tau_L}) \iff  \var(D \hat{\betab}_{\tau_L}) \leq \var(D \check{\betab}_{\tau_L}), \label{eq:appendix:aitkens_thm2}
		\end{align}
		\noindent
		where $D \in \rr^{m \times n \card(\tau_L)}$ is an arbitrary matrix. 
	\end{lemma}
	\begin{remark}
		Continuing to use the notation above for the model $\yb_{\tree} = \er \betab_{\tau_L} + \ub_{\tree}$ solely for notational convenience, note that if $\er$ has full column rank, then $\hat{\betab}_{\tau_L}$ is unique. We will use one implication of uniqueness of the GLS estimator repeatedly in the proofs below. Specifically, suppose we have an estimate $f(\yb_{\tree})$ for $\betab_{\tau_L},$ where $f(\cdot)$ is a function that might be expressed in a complex form that is difficult to reason about directly. However, also suppose that we are able to show that $f(\yb_{\tree})$ is the GLS estimator for $\betab(v)$ when the information set is given by the observations $\yb_{\tree},$ \textit{i.e.}, $f(\yb_{\tree})$ is unbiased, linear in $\yb_{\tree},$ and its variance matrix obtains the lower bound in \eqref{eq:appendix:aitkens_thm2}. As a result, uniqueness of the GLS estimator implies that $f(\yb_{\tree})$ can be expressed in the conceptually simple form given by $(\er^\top V_{\tree}^{-1} \er)^{-1}\er^\top V_{\tree}^{-1} \yb_{\tree}.$ 
		
		Also, we make use of the second inequality in \eqref{eq:appendix:aitkens_thm2} to reason about the GLS estimator for the histograms of vertices in levels above $L.$ For example, Aitken's Theorem implies that the GLS estimator for $\betab_{\gamma_L}$ based on the information set given by the observations $\yb_{\tree_v}$ is given by $ \hat{\betab}_{\gamma_L}= (\eg^\top V_{\tree_v}^{-1} \eg)^{-1}\eg^\top V_{\tree_v}^{-1} \yb_{\tree_v},$ so the second inequality in \eqref{eq:appendix:aitkens_thm2} implies that the GLS of $\betab(v)$ for this same information set is $\jlg\hat{\betab}_{\gamma_L}= \jlg (\eg^\top V_{\tree_v}^{-1} \eg)^{-1}\eg^\top V_{\tree_v}^{-1} \yb_{\tree_v}.$ \qed
	\end{remark}
	
	\begin{figure}[tbh]
		\centering
		\begin{tikzpicture}[
			level 1/.style = {sibling distance = 4.7cm}, 
			level 2/.style = {sibling distance = 3.7cm},
			]
			\node (g) {$v$}
			child { node (1) {$\gamma_C[0]$ }
				child { node (3) {$\gamma_L[0]$}}
				child { node (4) {$\gamma_L[1]$ }}
			}
			child { node (2) { $\gamma_C[1]$}
				child {node (5) {$\gamma_L[2]$}}
			};
		\end{tikzpicture}
		\caption{The sub-graph $\tree_v$ considered in Example \ref{eg:appendix:ec}.}
		\label{fig:appendix:eg_spine_intro}
	\end{figure}
	\begin{example} \label{eg:appendix:ec}
		Consider the sub-graph $\tree_v$ for $v\in\level(\tree, L-2)$ shown in Figure \ref{fig:appendix:eg_spine_intro}. In this case $\eg$ is given by 
		\[ \eg = \begin{bmatrix}
			S(v) & S(v) &  S(v) \\
			S(\gamma_C[0]) & S(\gamma_C[0]) & 0 \\
			0 & 0 & S(\gamma_C[1])\\
			S(\gamma_L[0]) & 0 & 0 \\
			0 & S(\gamma_L[1])  & 0 \\
			0 & 0 &  S(\gamma_L[2])  
		\end{bmatrix}.\]
		\noindent
		As described in Lemma \ref{lem:appendix:gauss_markov_thm} in more detail, the GLS estimator of $\betab_{\gamma_L}$ based on the observations in the sub-graph $\tree_v$ can be expressed as 
		\[\hat{\betab}_{\gamma_L} = (\eg^\top V_{\tree_v}^{-1} \eg)^{-1}\eg^\top V_{\tree_v}^{-1} \yb_{\tree_v}.\]
		
		Also, in this case $\jlc$ and $\jlg$ are given by
		\[\jlc = \begin{bmatrix}
			I_n & I_n & 0 \\
			0 & 0 & I_n
		\end{bmatrix} \textrm{ and } \jlg = \begin{bmatrix}
			I_n & I_n & I_n 
		\end{bmatrix}.\]   \qed
	\end{example}
	
	\subsection{Fine-to-Coarse Recursion}
	
	In the main result of this section, we show that, for each $v\in\tree,$ $\hat{\betab}(v|v-)$ is the BLUE for the information set given by the observations in $\tree_v.$ The following first two intermediate lemmas provide linear algebra results that will be helpful to prove the main result of this subsection. 
	
	\begin{lemma} \label{lem:appendix:var_beta_hat_g}
		The following equality holds
		\begin{align*}
			(S(v)^\top S(v) + (\jlg (\ec^\top \ec)^{-1} \jlg^\top )^{-1} )^{-1} = \jlg (\eg^\top \eg)^{-1} \jlg^\top.
		\end{align*}
	\end{lemma}
	\begin{proof}
		Given two non-singular matrices, $A,B\in\rr^{N\times N},$ a corollary of the Sherman–Morrison-Woodbury (SMW) lemma is that $(A+B)^{-1}=A^{-1}-(A+A B^{-1}A)^{-1},$ which implies
		\begin{align}
			& (S(v)^\top S(v) + (\jlg (\ec^\top \ec)^{-1} \jlg^\top )^{-1} )^{-1} \nonumber \\
			&= (S(v)^\top S(v))^{-1} - \left( S(v)^\top S(v) + S(v)^\top S(v)    \jlg (\ec^\top \ec)^{-1} \jlg^\top  S(v)^\top S(v) \right)^{-1}.  \label{eq:appendix:var_beta_hat_g}
		\end{align}
		Using $\check{S}$ as shorthand for $S(v)^\top S(v)$ to keep our notation concise in the remainder of the proof, a direct application of the SMW lemma to the second of the two terms in \eqref{eq:appendix:var_beta_hat_g} implies 
		\begin{align*}
			& (\check{S} + (\jlg (\ec^\top \ec)^{-1} \jlg^\top )^{-1} )^{-1} \\
			&= \check{S}^{-1}  - \check{S}^{-1} + \check{S}^{-1} \check{S} \jlg \left(    \ec^\top \ec + \jlg^\top \check{S} \check{S}^{-1} \check{S} \jlg \right)^{-1} \jlg^\top \check{S} \check{S}^{-1} \nonumber \\
			&= \jlg \left( \ec^\top \ec +  \jlg^\top S(v)^\top S(v)\jlg \right)^{-1} \jlg^\top = \jlg (\eg^\top \eg)^{-1} \jlg^\top.
		\end{align*}
	\end{proof}
	
	The next several lemmas will use the orthogonal projection matrix $M,$ which we define as
	\[ M= \begin{bmatrix}
		I_{m(v)} & 0 \\ 
		0 & \ec (\ec^\top \ec)^{-1} \jlg^\top \left(\jlg (\ec^\top \ec)^{-1} \jlg^\top\right)^{-1} \jlg (\ec^\top \ec)^{-1} \ec^\top
	\end{bmatrix},\]
	where $m(v)$ denotes the number of rows of $S(v).$
	
	\begin{lemma} \label{lem:appendix:inv_gram_mat_quad_equality}
		The matrix $M$ satisfies
		\begin{align}
			\jlg (\eg^\top M \eg)^{+} \jlg^\top  = \jlg (\eg^\top \eg)^{-1} \jlg^\top. \label{eq:appendix:inv_gram_mat_quad_equality}
		\end{align}
	\end{lemma}
	\begin{proof}
		Let $\kappa=\card(\gamma_L).$ Note that the definition of $M$ implies
		\begin{align*}
			& \jlg (\eg^\top M \eg)^{+} \jlg^\top = \jlg \left( \jlg^\top \left(S(v)^\top S(v) + (\jlg (\ec^\top \ec)^{-1} \jlg^\top)^{-1} \right)\jlg \right)^{+}  \jlg^\top \\
			& = \jlg \left(\mathbf{1}_{\kappa\times \kappa} \otimes \left(S(v)^\top S(v) + (\jlg (\ec^\top \ec)^{-1} \jlg^\top )^{-1}\right) \right)^{+} \jlg^\top 
		\end{align*}
		\noindent 
		Since $\mathbf{1}_{\kappa\times \kappa}^+ = \mathbf{1}_{\kappa\times \kappa}/\kappa^2$ and since the second matrix in the Kronecker product above is invertible by assumption \ref{assump:srank}, we have
		\begin{align*}
			&\jlg (\eg^\top M \eg)^{+} \jlg^\top  \\
			&=  \left( \mathbf{1}^\top_\kappa \otimes I_n  \right) \left(\mathbf{1}_{\kappa\times \kappa} / \kappa^2 \otimes \left(S(v)^\top S(v) + (\jlg (\ec^\top \ec)^{-1} \jlg^\top )^{-1}\right)^{-1} \right) \left( \mathbf{1}_\kappa \otimes I_n  \right) \\
			&=  \left(\mathbf{1}^\top_\kappa \mathbf{1}_{\kappa\times \kappa} \mathbf{1}_\kappa/ \kappa^2 \right) \otimes \left(S(v)^\top S(v) + (\jlg (\ec^\top \ec)^{-1} \jlg^\top )^{-1}\right)^{-1}  \\
			& =  \left(S(v)^\top S(v) + (\jlg (\ec^\top \ec)^{-1} \jlg^\top )^{-1}\right)^{-1},
		\end{align*}
		\noindent
		so Lemma \ref{lem:appendix:var_beta_hat_g} implies $\jlg (\eg^\top M \eg)^{+} \jlg^\top  = \jlg (\eg^\top \eg)^{-1} \jlg^\top.$
	\end{proof}
	
	The following lemma provides an alternative formulation for the GLS estimate of $\betab_g$ based on the observations associated with the vertices in $\tree_v,$ for the case in which $V_{\tree_v} = I.$
	
	\begin{lemma} \label{lem:appendix:hat_betab_prime}
		Suppose $V_{\tree_v} = I,$ and let  $\hat{\betab}_v' = \jlg (M \eg)^{+} \yb_{\tree_v}.$ Then, $\hat{\betab}_v'$ and  $\hat{\betab}_v = \jlg (\eg^\top \eg)^{-1} \eg^\top \yb_{\tree_v}$ are equal. 
	\end{lemma}
	\begin{proof}
		Under the assumption $V_{\tree_v} = I,$ the GLS estimator of $\betab_g = \jlg \betab_{\gamma_L}$ for the information set given by the observations in $\tree_v,$ is the GLS estimator based on the modeling equation $\yb_{\tree_v}=\eg \betab_{\gamma_L} + \ub_{\tree_v},$ \textit{i.e.},  $\hat{\betab}_v = \jlg \hat{\betab}_{\gamma_L} = \jlg (\eg^\top \eg)^{-1} \eg^\top  \yb_{\tree_v}.$  Consider instead the transformed model $M \yb_{\tree_v}=M \eg \betab_{\gamma_L} + M \ub_{\tree_v}.$ Since $M$ is an orthogonal projection matrix, it is generally singular, so the standard GLS estimator is not necessarily defined in this case. In other words, this transformed model is not equivalent to the original modeling equation because $M$ may be rank deficient, so in general it cannot be used to derive the GLS estimator of $\betab_{\gamma_L}.$ However, in this proof we show that the solution to this transformed model still provides the GLS estimator of $\betab_{v}$ for the relevant information set, and thus $\hat{\betab}_v = \hat{\betab}_v'.$ We show $\hat{\betab}_v'$ is the GLS estimator by showing that it is linear, unbiased, and that its variance matrix is identical to that of $\hat{\betab}_v.$ 
		
		To do this, we first prove an intermediate equality. Let $A=M \eg$ and let $Q$ be defined as the block matrix
		\begin{align*}
			& Q = \begin{bmatrix}
				(S(v)^\top S(v))^{-1} S(v)^\top & 0
			\end{bmatrix}.
		\end{align*}
		The definitions of $M$ and $Q$ imply $Q A = Q M \eg = (S(v)^\top S(v))^{-1} S(v)^\top F_g = \jlg,$ which implies $\jlg$ is in the row space of $A.$ Thus, the property of the pseudoinverse that $A^{+} A$ is the orthogonal projection matrix onto the row space of $A$ implies $A^{+} A \jlg^\top = \jlg^\top.$\footnote{This can be shown by replacing $A$ and $A^+$ in $A^{+} A$ with their singular value decompositions and then simplifying.} Since orthogonal projection matrices are symmetric, this implies 
		\begin{align}
			\jlg A^{+} A  = \jlg. \label{eq:appendix:jlg_a_pinv_a_is_jlg}
		\end{align}
		
		The solution to the transformed model that we consider is
		\begin{align*}
			\hat{\betab}_{\gamma_L}' = \argmin_{\betab\in\rr^{n \card(\gamma_L)}} \lVert \eg \betab - \yb_{\tree_v} \rVert_2^2  = (M \eg)^{+} \yb_{\tree_v} = A^+ \yb_{\tree_v}.\label{eq:appendix:aitkens_thm_specific_case}
		\end{align*}
		Clearly $\hat{\betab}_{v}'=\jlg \hat{\betab}_{\gamma_L}'=\jlg A^+ \yb_{\tree_v}$ is linear in $\yb_{\tree_v}.$ Second, to show $\hat{\betab}_{v}'$ is unbiased, we have
		\begin{align*}
			E(\jlg \hat{\betab}_{\gamma_L}')= \jlg A^+ E(\yb_{\tree_v}) = \jlg A^+ \left(\eg \betab_{\gamma_L} + E(\ub_{\tree_v})\right) = \jlg A^+ \eg \betab_{\gamma_L}.
		\end{align*}
		The property of the pseudoinverse that $A^+=(A^\top A)^{+} A^\top$ and the fact that $M^2 =M$ imply
		\begin{align*}
			\jlg A^+ \eg \betab_{\gamma_L} =&  \jlg (\eg^\top M^\top M\eg)^{+} \eg^\top M^\top  \eg \betab_{\gamma_L} = \jlg \left( (A^\top A)^{+} A^\top \right) M\eg \betab_{\gamma_L} \\
			=& \jlg A^{+}  A \betab_{\gamma_L}.
		\end{align*}
		Thus, property \eqref{eq:appendix:jlg_a_pinv_a_is_jlg} implies $E(\jlg \hat{\betab}_{\gamma_L}') = \jlg \betab_{\gamma_L}  = \betab_g,$ so $\hat{\betab}_{v}'$ is unbiased. Third, after deriving the variance matrix and simplifying, we have
		\begin{align*}
			\var(\hat{\betab}'_{v}) =&  E((\hat{\betab}'_{v} -\betab_{v}) (\hat{\betab}_{v}^{\prime}-\betab_{v})^\top) = \jlg A^{+} E( \ub_{\tree_v} \ub_{\tree_v}^\top) A^{+^\top} \jlg^\top = \jlg A^{+} A^{+^\top} \jlg^\top  \\
			=& \jlg (A^\top A)^{+} \jlg^\top,
		\end{align*}
		so Lemma \ref{lem:appendix:inv_gram_mat_quad_equality} implies $\var(\hat{\betab}'_{v}) = \var(\hat{\betab}_{v}),$ which implies the result.  
	\end{proof}
	
	The next lemma is used in the induction step in the proof of the main result of this subsection, \textit{i.e.}, Theorem \ref{thm:appendix:g_cond_g_minus_is_blue_for_any_level}.
	
	\begin{lemma} \label{lem:appendix:g_cond_g_minus_is_blue_for_level_L_minus_one}
		Suppose each $\hat{\betab}(c|c-) \in \{\hat{\betab}(c|c-)\}_{c\in\gamma_C}$ is equal to the GLS estimator for the information set given by the observations in $\tree_c.$ Then, $\hat{\betab}(v | v-)$ is the GLS estimator for the information set given by the observations in $\tree_v.$ 
	\end{lemma}
	\begin{proof}
		To keep the notation concise, in this proof we will assume that $V_{\tree_v}$ is an identity matrix. This is without loss of generality because, if it does not hold for the original linear model $\yb_{\tree_v} = \eg \betab_{\gamma_L} + \ub_{\tree_v},$ it is possible to formulate an equivalent model that satisfies this assumption by left multiplying $\var(\ub_{\tree_v})^{-1/2}$ on both sides of this equality. In other words, it is always possible to redefine $\eg,\yb_{\tree_v},$ and $ \ub_{\tree_v}$ as $\var(\ub_{\tree_v})^{-1/2}\eg,\var(\ub_{\tree_v})^{-1/2}\yb_{\tree_v},$ and $ \var(\ub_{\tree_v})^{-1/2}\ub_{\tree_v},$ respectively, to ensure this assumption holds without impacting the minimizer 
		\[\hat{\betab}_{\gamma_L}=\argmin_{\betab_{\gamma_L}\in\rr^{n \card(\gamma_L)}} \lVert  V_{\tree_v}^{-1/2} (\eg \betab_{\gamma_L} - \yb_{\tree_v}) \rVert_2^2.\]
		
		%\begin{align}
		%    \hat{\betab}_{\gamma_L} = \argmin_{\betab\in\rr^{n \card(\gamma_L)}} \lVert  V_{\tree}^{-1/2} (\eg \betab - \yb_{\tree}) \rVert_2^2  = (\eg^\top V_{\tree}^{-1} \eg)^{-1}\eg^\top V_{\tree}^{-1} \yb_{\tree} 
		%\end{align}
		
		Lemma \ref{lem:appendix:gauss_markov_thm} implies that, for the relevant information set, $\hat{\betab}_{\gamma_L} = (\eg^\top \eg)^{-1} \eg^\top \yb_{\tree_v}$ is the GLS estimator of $\betab_{\gamma_L}$ the relevant information set. So, to prove this result we need to show that $\hat{\betab}(v | v-) = \jlg \hat{\betab}_{\gamma_L}= \jlg (\eg^\top \eg)^{-1} \eg^\top \yb_{\tree_v}.$ To do so, we will begin by translating the components of \eqref{eq:up_g_g_minus2} into the matrix notation of this appendix. Specifically, first let $\hat{\betab}_{\gamma_C|\gamma_C-}=\stack(\{\hat{\betab}(c|c-)\}_{c\in\gamma_C}).$ Since each $\hat{\betab}(c|c-) \in \{\hat{\betab}(c|c-)\}_{c\in\gamma_C}$ is the GLS estimator for the information set given by the observations in $\tree_c,$ Aitken's Theorem implies that $\hat{\betab}_{\gamma_C|\gamma_C-}= \jlc (\ec^\top \ec)^{-1}\ec^\top \yb_\gamma,$ and thus 
		\begin{align*}
			\hat{\betab}(v|\children(v)-) &=  \jcg \jlc (\ec^\top \ec)^{-1}\ec^\top \yb_\gamma\\
			&=  \jlg (\ec^\top \ec)^{-1}\ec^\top \yb_\gamma,
		\end{align*}
		which has a variance matrix given by  $\var(\hat{\betab}(v|\children(v)-))=  \jlg (\ec^\top \ec)^{-1} \jlg^\top.$ Second, equations \eqref{eq:up_init1} and \eqref{eq:up_init2} can be written as, $\hat{\betab}(v|v) = (S(v)^\top S(v))^{-1} S(v)^\top \yb(v)$ and $\var(\hat{\betab}(v|v)) = (S(v)^\top S(v))^{-1},$ respectively. Third, after translating \eqref{eq:up_g_g_minus1} into the notation introduced in this appendix and then applying Lemma \ref{lem:appendix:var_beta_hat_g}, we have 
		\begin{align*}
			\var(\hat{\betab}(v|v-)) = ( S(v)^\top S(v) + (\jlg (\ec^\top \ec)^{-1} \jlg^\top )^{-1} )^{-1} = \jlg (\eg^\top \eg)^{-1} \jlg^\top. 
		\end{align*}
		
		After substituting each of these components into the definition of $\hat{\betab}(v|v-)$ in \eqref{eq:up_g_g_minus2}, we have
		\begin{align*}
			\hat{\betab}(v|v-) & =\jlg (\eg^\top \eg)^{-1} \jlg^\top \begin{bmatrix}
				S(v)^\top  &  \left(\jlg (\ec^\top \ec)^{-1} \jlg^\top\right)^{-1} \jlg (\ec^\top \ec)^{-1} \ec^\top
			\end{bmatrix} \yb_{\tree_v},
		\end{align*}
		so after using Lemma \ref{lem:appendix:inv_gram_mat_quad_equality} and the definition of $M$ in each of the next two equalities, respectively, we have 
		\begin{align}
			\hat{\betab}(v|v-) & = \jlg (\eg^\top M \eg)^{+}  \begin{bmatrix}
				\jlg^\top S(v)^\top  &  \jlg^\top \left(\jlg (\ec^\top \ec)^{-1} \jlg^\top\right)^{-1} \jlg (\ec^\top \ec)^{-1} \ec^\top
			\end{bmatrix} \yb_{\tree_v}  \nonumber \\
			& = \jlg (\eg^\top M \eg)^{+}  \eg^\top M \yb_{\tree_v} = \jlg (M \eg)^{+} \yb_{\tree_v}.  \label{eq:appendix:intermediate_beta_hat_g_minus} 
		\end{align}
		Since Lemma \ref{lem:appendix:hat_betab_prime} implies \eqref{eq:appendix:intermediate_beta_hat_g_minus} is equal to the GLS estimator of $\betab_g$ for the relevant information set, this implies the final result.
	\end{proof}
	
	The next result provides the final result of this subsection on the optimality of the estimates computed in the fine-to-coarse recursion for a specific information set. 
	
	\begin{theorem}
		\label{thm:appendix:g_cond_g_minus_is_blue_for_any_level}
		For any $v\in\tree,$ the GLS estimator of $\betab(v)$ based on the information set given by the observations in $\tree_v$ is equal to $\hat{\betab}(v | v-).$
	\end{theorem}
	\begin{proof}
		Aitken's Theorem implies that, for each leaf vertex $v\in\level(\tree, L),$  $\hat{\betab}(v|v-)$ is the GLS estimator for the information set consisting of the observations in $\tree_v.$ If each $\hat{\betab}(c|c-) \in \{\hat{\betab}(c|c-)\}_{c\in\gamma_C}$ is equal to the GLS estimator for the information set given by the observations in $\tree_c,$ Lemma \ref{lem:appendix:g_cond_g_minus_is_blue_for_level_L_minus_one} implies that $\hat{\betab}(v | v-)$ is the GLS estimator for the information set given by the observations in $\tree_v,$ so this result follows by induction.
	\end{proof}
	
	\subsection{Coarse-to-Fine Recursion}\label{sec:appendix:tdp}
	
	In this subsection we will prove that the full-information GLS estimator of $\betab_{\tau_L},$ which is equal to $\hat{\betab}_{\tau_L}= (\er^\top V_\tree^{-1}\er)^{-1}\er^\top V_\tree^{-1}\yb_\tree,$ is given by $\stack(\{\tilde{\betab}(v)\}_{v\in\level(\tree,L)}).$ Before doing so, it will be helpful to describe the main implication of our assumption that the rows and columns of design matrices are ordered with respect to the parent vertex $v,$ as defined in more detail in Appendix \ref{sec:appendix:notation}. Specifically, this assumption simplifies our derivations slightly because it allows us to write $\er$ as
	
	\begin{align}
		&\er = \begin{bmatrix}
			F_1 & F_2 \\
			0 & \ec
		\end{bmatrix}, \label{eq:appendix:ordered_wrt_vert1}
	\end{align} 
	\noindent
	where $F_1,F_2$ are defined so that the rows of $\begin{bmatrix} F_1 & F_2 \end{bmatrix}$ encode the linear queries of vertices that are ancestors of $v$ and $v$ itself first and then the linear queries of the remaining vertices in $\{v\in \tree\}/\gamma.$ Example \ref{eg:appendix:spine} provides this ordering for the rooted tree depicted in Figure \ref{fig:appendix:eg_spine}.
	
	\begin{figure}[tbh]
		\centering
		\begin{tikzpicture}[
			level 1/.style = {sibling distance = 4.7cm}, 
			level 2/.style = {sibling distance = 3.7cm},
			level 3/.style = {sibling distance = 2.7cm},
			]
			\node (r) {$r$}
			child { node (1) {$c_1 $ }
				child { node (3) {$v$}
					child { node (6) {$\gamma_L[0]$ }}
					child { node (7) {$\gamma_L[1]$ } }}
				child { node (4) {$c_4$ }
					child {node (8) {$c_8$}}
				}
			}
			child { node (2) { $c_2$}
				child {node (5) {$c_5$}
					child { node (9) {$c_9$} }
					child {node (10) {$c_{10}$} }
				}
			};
		\end{tikzpicture}
		\caption{The rooted tree considered in Example \ref{eg:appendix:spine}.}
		\label{fig:appendix:eg_spine}
	\end{figure}
	\begin{example}\label{eg:appendix:spine}
		This example considers how to represent the design matrix $\er$ so that its rows are ordered with respect to vertex $v,$ for the tree provided in Figure \ref{fig:appendix:eg_spine}. After moving each row corresponding to a query that can be expressed by a sum that includes one or more detailed histogram cell counts of vertex $v$ to the initial rows of $\er,$ and then moving rows corresponding to queries of vertices that are descendants of vertex $v$ to the final rows of $\er,$ $\er$ can be defined according to \eqref{eq:appendix:ordered_wrt_vert1} such that
		\begin{align*}
			F_1 = \begin{bmatrix}
				S(r) & S(r) & S(r) \\
				S(c_1) & 0 & 0 \\
				0 & 0 & 0 \\
				0 & S(c_2) & S(c_2) \\
				S(c_4) & 0 & 0 \\
				0 & S(c_5) & S(c_5) \\
				S(c_8) & 0 & 0 \\
				0 & S(c_9) & 0 \\
				0 & 0 & S(c_{10}) \\
			\end{bmatrix},  \; F_2 = \begin{bmatrix}
				S(r) & S(r) \\
				S(c_1) & S(c_1) \\
				S(v) & S(v) \\
				0 & 0 \\
				0 & 0 \\
				0 & 0 \\
				0 & 0 \\
				0 & 0 \\
				0 & 0 \\
			\end{bmatrix}, \textrm{ and } \ec = \begin{bmatrix}
				S(\gamma_L[0]) & 0 \\
				0 & S(\gamma_L[1])
			\end{bmatrix}.
		\end{align*}
		\noindent
		Likewise, if the rows of $\er$ are assumed to be in this order, it also implies that $\yb_\tree$ is defined so that its elements are ordered as 
		\[(\yb(r)^\top, \yb(c_1)^\top, \yb(v)^\top  \yb(c_2)^\top, \yb(c_4)^\top,\yb(c_5)^\top, \yb(c_8)^\top, \yb(c_9)^\top,\yb(c_{10})^\top,\yb(\gamma_L[0])^\top, \yb(\gamma_L[1])^\top)^\top, \] 
		\noindent
		and likewise that $\ub_\tree$ is defined so that its block elements are in this same vertex order. \qed   
	\end{example}
	
	The following lemma provides three basic properties of the matrices $F_1, F_2, $ and $\ec$ that we will use in the proofs below. 
	
	\begin{lemma} \label{lem:appendix:properties_of_F_r}
		If $\er$ is ordered with respect to vertex $v,$ then the matrices $F_1$ and $\ec$ have full column rank, and there exists a block diagonal matrix $C$ such that
		\begin{align*}
			F_2 = C \begin{bmatrix}
				\jlg \\
				\vdots \\
				\jlg \\
				0 \\
				\vdots \\
				0
			\end{bmatrix}.
		\end{align*}
	\end{lemma}
	\begin{proof}
		$\ec$ has full column rank because, possibly after permuting the rows and columns of $\ec,$ the rows of $\ec$ include the rows of 
		\begin{align*}
			\begin{bmatrix}
				S(\gamma_L[0]) & 0 & \cdots \\
				0 &  S(\gamma_L[1]) &  \\
				\vdots &   & \ddots 
			\end{bmatrix},
		\end{align*}
		which has full column rank by Assumption \eqref{assump:srank}. Likewise, by the same assumption, $F_1$ has full column rank because, again, 
		possibly after permuting, $F_1$ includes the rows of
		\begin{align*}
			\begin{bmatrix}
				S(\left(\level(\tree,L) / \gamma_L\right)[0]) & 0 & \cdots \\
				0 &  S(\left(\level(\tree,L) / \gamma_L\right)[1]) &  \\
				\vdots &   & \ddots 
			\end{bmatrix}.
		\end{align*}
		
		The property for $F_2$ holds for $C$ defined as the block diagonal matrix with diagonal blocks given by $\{S((\tree/\gamma)[i])\}_{i}.$ In other words, if the number of ancestors of vertex $v$ is denoted by $m,$ then
		\begin{align*}
			F_2 = \begin{bmatrix}
				S((\tree/\gamma)[0]) &  \cdots & S((\tree/\gamma)[0])  \\
				\vdots  &  &  \vdots  \\
				S((\tree/\gamma)[m]) &  \cdots &  S((\tree/\gamma)[m])  \\
				0 & \cdots & 0\\
				\vdots  &  &  \vdots  \\
				0 & \cdots & 0\\
			\end{bmatrix} = \begin{bmatrix}
				S((\tree/\gamma)[0]) & 0 & \cdots \\
				0 &  S((\tree/\gamma)[1]) &  \\
				\vdots &   & \ddots 
			\end{bmatrix} \begin{bmatrix}
				I & \cdots & I \\
				\vdots  &  &   \vdots  \\
				I & \cdots & I \\
				0 & \cdots & 0\\
				\vdots  &  &  \vdots  \\
				0 & \cdots & 0\\
			\end{bmatrix} = C \begin{bmatrix}
				\jlg \\
				\vdots \\
				\jlg \\
				0 \\
				\vdots \\
				0
			\end{bmatrix}.
		\end{align*}
	\end{proof}
	
	The next two lemmas are used to prove the induction step in the proof of the main result.
	
	\begin{lemma} \label{lem:appendix:beta_check_eq}
		Suppose $V_\tree=I,$ and let 
		\begin{align*}
			& R=(\ec^\top \ec)^{-1} \jlg^\top  (\jlg (\ec^\top \ec)^{-1} \jlg^\top)^{-1}  \jlg, \\
			& W=I-R, \textrm{ and } \\
			& \check{\betab}_{\gamma_L} = W (\ec^\top \ec)^{-1}  \ec^\top \yb_\gamma+ R \jrl (\er^\top \er)^{-1} \er^\top \yb_\tree.
		\end{align*}
		\noindent Then for $ \hat{\betab}_{\tau_L}= (\er^\top \er)^{-1} \er^\top \yb_\tree,$ we have 
		\[\check{\betab}_{\gamma_L} =\jrl \hat{\betab}_{\tau_L}.\]
	\end{lemma}
	\begin{proof}
		This proof involves a variety of Gram matrices, and a few notes about which of these matrices are full rank will be helpful. Specifically, the first two properties of Lemma \ref{lem:appendix:properties_of_F_r} imply that $F_1^\top F_1$ and $ \ec^\top\ec$ have full rank, which in turn implies that $ \er^\top \er $ also has full rank. However, generally $F_2^\top F_2$ does not have full rank, since the column rank of $F_2$ is $n,$ \textit{i.e.}, the dimension of $\betab(v)$ for each $v\in \tree,$ but $F_2$ has $n \card(\gamma_L) \geq n$ columns.
		
		We will start by proving two intermediate equalities. First, using the properties that $R^2 = R$ and $R W=0,$ we have
		\begin{align*}
			R \check{\betab}_{\gamma_L} = R(W (\ec^\top \ec)^{-1}  \ec^\top \yb_\gamma+ R \jrl (\er^\top \er)^{-1} \er^\top \yb_\tree) = R \jrl (\er^\top \er)^{-1} \er^\top \yb_\tree,
		\end{align*}
		and, since $\hat{\betab}_{\tau_L} =(\er^\top \er)^{-1} \er^\top \yb_\tree$ under the assumption that $V_\tree=I,$ this implies the first intermediate equality
		\begin{align}
			R \check{\betab}_{\gamma_L} = R \jrl \hat{\betab}_{\tau_L} . \label{eq:appendix:td_nts1}
		\end{align}
		
		The second intermediate equality we will prove is 
		\begin{align}
			& W\check{\betab}_{\gamma_L} = W \jrl \hat{\betab}_{\tau_L} . \label{eq:appendix:td_nts2}
		\end{align}
		To establish this equality, note that $W \jrl \hat{\betab}_{\tau_L} $ is the GLS estimator for $W\jrl \betab_{\tau_L}.$ Since the GLS estimator is unique, we will prove this equality by showing that  $W\check{\betab}_{\gamma_L}$ is also the GLS estimator of $W\jrl \betab_{\tau_L}.$ To do this, note that $W\check{\betab}_{\gamma_L}$ is linear in $\yb_\tree.$ Second, $W\check{\betab}_{\gamma_L}$ is unbiased because 
		\begin{align*}
			E(W \check{\betab}_{\gamma_L})& = W (\ec^\top \ec)^{-1}  \ec^\top E(\yb_\gamma)  =W (\ec^\top \ec)^{-1}  \ec^\top E(\ec \jrl \betab_{\tau_L}  + \ub_{\gamma}) \\
			&  = W (\ec^\top \ec)^{-1}  \ec^\top \ec \jrl \betab_{\tau_L}  + W (\ec^\top \ec)^{-1}  \ec^\top E(\ub_{\gamma})) = W \jrl \betab_{\tau_L}.
		\end{align*}
		\noindent
		To show $\var(W\check{\betab}_{\gamma_L}) = \var(W \jrl \hat{\betab}_{\tau_L} ),$ suppose the design matrix $\er$ is ordered with respect to vertex $v.$ This allows us to write $\var(W \jrl \hat{\betab}_{\tau_L} )$ concisely as
		
		\[ \var(W \jrl \hat{\betab}_{\tau_L} ) = W \jrl \left[
		\begin{array}{cc}
			F_1^\top F_1  & F_1^\top F_2  \\
			F_2^\top F_1 & F_2^\top F_2 + \ec^\top \ec
		\end{array}
		\right] ^{-1} \jrl^\top  W^\top.\]
		Note that multiplying the inner inverse matrix by $\jrl$ on the left and $\jrl^\top$ on the right in the expression above amounts to isolating the bottom right block of this inverse matrix. Thus, the formula for the inverse of a $2\times 2$ block matrix implies
		\begin{align}
			\var(W \jrl \hat{\betab}_{\tau_L} ) = W \left( \ec^\top \ec + F_2^\top \left( I-   F_1  (F_1^\top F_1)^{-1}   F_1^\top \right) F_2 \right)^{-1} W^\top. \label{eq:appendix:intmdt_var_beta_hat}
		\end{align}
		
		To make our derivations more concise, let $M_1$ be defined as the orthogonal projection matrix onto the null space of $F_1,$ \textit{i.e.}, $M_1 = I-F_1  (F_1^\top F_1)^{-1}   F_1^\top.$ Since $F_2^\top M_1 F_2$ is not generally full rank, the SMW lemma cannot be used to directly to simplify \eqref{eq:appendix:intmdt_var_beta_hat}, so instead we will use the variant of this lemma provided by \cite{henderson1981deriving}, which does not require this matrix to be non-singular. Specifically, this lemma implies
		\begin{align*}
			\var(W \jrl \hat{\betab}_{\tau_L} ) &= W \left( \ec^\top \ec + F_2^\top M_1 F_2 \right)^{-1} W^\top \\
			& = W \left( \left( \ec^\top \ec\right)^{-1} - \left(I+ \left( \ec^\top \ec\right)^{-1}F_2^\top   M_1 F_2   \right)^{-1}    \left( \ec^\top \ec\right)^{-1}   F_2^\top M_1 F_2   \left( \ec^\top \ec\right)^{-1} \right) W^\top.
		\end{align*}
		This, along with the fact that $\var(W\check{\betab}_{\gamma_L}) = \var(W (\ec^\top \ec)^{-1}  \ec^\top \yb_\gamma) = W (\ec^\top \ec)^{-1} W^\top,$ implies that \eqref{eq:appendix:td_nts2} holds if and only if
		\begin{align}
			& \var(W\check{\betab}_{\gamma_L}) - \var(W \jrl \hat{\betab}_{\tau_L} )=0  \nonumber \\
			\iff &  W \left(I+ \left( \ec^\top \ec\right)^{-1}F_2^\top   M_1 F_2   \right)^{-1}    \left( \ec^\top \ec\right)^{-1}   F_2^\top M_1 F_2   \left( \ec^\top \ec\right)^{-1}  W^\top = 0 \nonumber \\
			\iff & T F_2   \left( \ec^\top \ec\right)^{-1}  W^\top = 0, \label{eq:appendix:td_is_zero}
		\end{align}
		\noindent
		where 
		\[T=  W \left(I+ \left( \ec^\top \ec\right)^{-1}F_2^\top   M_1 F_2   \right)^{-1}    \left( \ec^\top \ec\right)^{-1}   F_2^\top M_1.\]
		To show \eqref{eq:appendix:td_is_zero} holds, note that the definition of $W$ and the final property provided by Lemma \ref{lem:appendix:properties_of_F_r} imply that there exists a matrix $C$ such that 
		\begin{align*}
			& T F_2   \left( \ec^\top \ec\right)^{-1}  W^\top  = T C \begin{bmatrix}
				\jlg \\
				\vdots \\
				\jlg \\
				0 \\
				\vdots
			\end{bmatrix} \left( \ec^\top \ec\right)^{-1} \left(I - \jlg^\top  (\jlg (\ec^\top \ec)^{-1} \jlg^\top)^{-1}  \jlg (\ec^\top \ec)^{-1} \right) \\
			& = T C \begin{bmatrix}
				I \\
				\vdots \\
				I  \\
				0 \\
				\vdots 
			\end{bmatrix} \left(\jlg \left( \ec^\top \ec\right)^{-1} -  \jlg (\ec^\top \ec)^{-1} \right) = 0.
		\end{align*}
		Since this implies \eqref{eq:appendix:td_is_zero} holds, $W\check{\betab}_{\gamma_L}$ is the GLS estimator of $W\jrl \betab_{\tau_L},$ which in turn implies equality \eqref{eq:appendix:td_nts2}).
		
		The final result follows from summing over equalities \eqref{eq:appendix:td_nts1} and \eqref{eq:appendix:td_nts2}; since $R+W=I,$ we have 
		\begin{align*}
			& R\check{\betab}_{\gamma_L} + W\check{\betab}_{\gamma_L} = R \jrl \hat{\betab}_{\tau_L}  +  W \jrl \hat{\betab}_{\tau_L} \implies \check{\betab}_{\gamma_L} = \jrl \hat{\betab}_{\tau_L} .
		\end{align*}
	\end{proof}
	
	\begin{lemma} \label{lem:appendix:beta_tilde_c_is_blue_for_level_l}
		Suppose that $l<L,$ $v\in\level(\tree,l),$ $\tilde{\betab}(v)$ is the full-information GLS estimator, and that, for each $c\in\children(v),$ $\hat{\betab}(c|c-)$ is the GLS estimator for the information set given the observations in $\tree_C.$ Then, for each $c\in\children(v),$ $\tilde{\betab}(c)$ is the full-information GLS estimator. 
	\end{lemma}
	\begin{proof}
		Throughout the proof we assume that $\var(\ub(v))=I$ for each $v\in\tree,$ which is without loss of generality by the same logic given in the first paragraph of Lemma \ref{lem:appendix:g_cond_g_minus_is_blue_for_level_L_minus_one}. 
		
		We will start with translating equations \eqref{eq:down1} and \eqref{eq:A_of_c_def} into the notation introduced in this appendix, \textit{i.e.}, we will derive both $\hat{\betab}_{c | c-} = \stack(\{\hat{\betab}(c|c-) \}_{c\in\gamma_C})$ and $A_c =\stack(\{A(c)\}_{c\in\gamma_C}).$ The assumptions in the statement of the Lemma imply
		\begin{align*}
			& \hat{\betab}_{c | c-}  = \jlc (\ec^\top \ec)^{-1} \ec^\top \yb_\gamma \textrm{  and} \\
			& \tilde{\betab}(v) = \jrg (\er^\top \er)^{-1} \er^\top \yb_\tree,
		\end{align*}
		so \eqref{eq:A_of_c_def} can be derived using
		\begin{align*}
			\var(\hat{\betab}_{c | c-}) &= \begin{bmatrix}
				\var(\hat{\betab}(\gamma_C[0]|\gamma_C[0]-)) & 0 & \cdots \\
				0 & \var(\hat{\betab}(\gamma_C[1]|\gamma_C[1]-)) & \\
				\vdots &  & \ddots \\
			\end{bmatrix} = \jlc (\ec^\top \ec)^{-1} \jlc^\top \implies \\
			A_c &= \var(\hat{\betab}_{c | c-}) \jcg^\top (\jcg \var(\hat{\betab}_{c | c-}) \jcg^\top )^{-1} \\
			& = \jlc (\ec^\top \ec)^{-1} \jlg^\top  (\jlg (\ec^\top \ec)^{-1} \jlg^\top )^{-1},
		\end{align*}
		where the final equality used the fact that $\jlg=\jcg\jlc.$ Likewise, to translate \eqref{eq:down1} to the notation in this appendix, we have
		\begin{align}
			\tilde{\betab}_c & = \hat{\betab}_{c | c-} + A_c \left( \tilde{\betab}(v) - \jcg \hat{\betab}_{c | c-} \right) \nonumber \\
			& = \jlc (\ec^\top \ec)^{-1} \ec^\top \yb_\gamma + A_c \left( \jrg (\er^\top \er)^{-1} \er^\top \yb_\tree - \jcg \jlc (\ec^\top \ec)^{-1} \ec^\top \yb_\gamma \right).  \label{eq:appendix:betab_c_1}
		\end{align} 
		
		After substituting in the definition of $A_c$ into \eqref{eq:appendix:betab_c_1} and simplifying using the definitions of $R,W,$ and $\check{\betab}_{\gamma_L}$ provided in Lemma \ref{lem:appendix:beta_check_eq}, we have
		
		\begin{align*}
			\tilde{\betab}_c =&  \jlc (\ec^\top \ec)^{-1} \ec^\top \yb_\gamma - \jlc R (\ec^\top \ec)^{-1} \ec^\top \yb_\gamma \\
			& + \jlc (\ec^\top \ec)^{-1} \jlg^\top (\jlg (\ec^\top \ec)^{-1} \jlg^\top )^{-1} \jrg (\er^\top \er)^{-1} \er^\top \yb_\tree  \\
			=&  \jlc \big( (I - R ) (\ec^\top \ec)^{-1} \ec^\top \yb_\gamma \\
			& + (\ec^\top \ec)^{-1} \jlg^\top (\jlg (\ec^\top \ec)^{-1} \jlg^\top )^{-1}   \jlg \jrl (\er^\top \er)^{-1} \er^\top \yb_\tree\big)  \\
			=&  \jlc \left( \left(I - R \right) (\ec^\top \ec)^{-1} \ec^\top \yb_\gamma  + R \jrl (\er^\top \er)^{-1} \er^\top \yb_\tree\right)  \\
			=& \jlc \left( W (\ec^\top \ec)^{-1}  \ec^\top \yb_\gamma+ R \jrl (\er^\top \er)^{-1} \er^\top \yb_\tree \right) = \jlc \check{\betab}_{\gamma_L}.
		\end{align*}
		Thus, Lemma \ref{lem:appendix:beta_check_eq} implies that 
		\[\tilde{\betab}_c = \jlc \check{\betab}_{\gamma_L} = \jlc\jrl (\er^\top \er)^{-1} \er^\top \yb_\tree= \jrc (\er^\top \er)^{-1} \er^\top \yb_\tree=\jrc\hat{\betab}_{\tau_L},\]
		which is the full-information GLS estimator of $\betab_{\gamma_C}$ by Aitken's Theorem. 
	\end{proof}
	
	The following theorem on the optimality of $\tilde{\betab}(v)$ for each $v\in\tree$ is our main result. 
	
	\begin{theorem} \label{thm:appendix:beta_tilde_is_blue}
		If Assumptions \ref{assump:spine_struct}-\ref{assump:u_iid} hold, then for each $v\in\tree$  $\tilde{\betab}(v),$ as defined in Section \ref{sec:two_pass_est}, and the value of $\tilde{\betab}_{H,\qb}$ returned from Algorithm \ref{alg:est_ci} is the full-information GLS estimator for $\betab(v).$ 
	\end{theorem}
	\begin{proof}
		Theorem \ref{thm:appendix:g_cond_g_minus_is_blue_for_any_level} implies that, for each $c\in\tree,$ $\hat{\betab}(c|c-)$ is the GLS estimator for the information set given the observations in $\tree_c.$ Also, when the parent vertex $v\in\tree$ is the root of $\tree,$ Theorem \ref{thm:appendix:g_cond_g_minus_is_blue_for_any_level} also implies $\tilde{\betab}(v)$ is the full-information GLS estimator. 
		
		For a parent vertex $v\in\tree,$ if $\tilde{\betab}(v)$ is the full-information GLS estimator, Lemma \ref{lem:appendix:beta_tilde_c_is_blue_for_level_l} implies that $\tilde{\betab}(c)$ is the full-information GLS estimator for each $c\in\children(v),$ so this result follows by induction.  
	\end{proof}
	
	\section{Closed-Form Finite Sample Distribution of the GLS Estimator} \label{sec:appendix:beta_tilde_is_normal}
	
	The CI estimator proposed here is statistically valid even if the sample size is finite, when the errors are Gaussian, which is proved in the following Theorem. 
    
    One can also show that the GLS estimator is statistically valid without a normality assumption, at least asymptotically as the sample size, \textit{i.e.}, the dimension of $\yb$ in \eqref{eq:yb_def_intro}, diverges, using the central limit theorem; for more detail, see \cite{greene2003econometric}. In the context of formally private mechanisms, note that care must be taken in the context of asymptotic results because observing additional noisy measurements requires expending a higher privacy loss budget (PLB). However, in cases in which relatively few (non-Gaussian) mean-zero errors are used within a formally private matrix mechanism, data curators may still view additional PLB expenditures to support more accurate inferences on the uncertainty introduced by disclosure limitation methods to be worthwhile.

	\begin{theorem}\label{thm:appendix:beta_tilde_is_normal}
		If Assumptions \ref{assump:spine_struct}-\ref{assump:normality} hold, then $\tilde{\betab}_{H,\qb}$ is normally distributed, and the $1-\alpha$ CI of $\betab_{H,\qb}$ output from Algorithm \ref{alg:est_ci} is statistically valid.
	\end{theorem}
	\begin{proof}
		Theorem \ref{thm:appendix:beta_tilde_is_blue} implies 
		\begin{align*}
			\tilde{\betab}_{H,\qb} &= (\boldsymbol{h}\otimes \qb)^\top (\er^\top V_\tree^{-1} \er)^{-1}  \er^\top V_\tree^{-1} \yb_\tree  =  (\boldsymbol{h}\otimes \qb)^\top (\er^\top V_\tree^{-1} \er)^{-1}  \er^\top V_\tree^{-1} \left( \er \betab_{\tau_L} + \ub_\tree \right) \\
			&= \betab_{H,\qb} + (\boldsymbol{h}\otimes \qb)^\top (\er^\top V_\tree^{-1} \er)^{-1}  \er^\top V_\tree^{-1}  \ub_\tree 
		\end{align*}
		Since a linear combination of normally distributed random variables is itself normally distributed, this implies
		\begin{align*}
			\tilde{\betab}_{H,\qb} \sim N\left( \betab_{H,\qb}, (\boldsymbol{h}\otimes \qb)^\top  (\er^\top V_\tree^{-1} \er)^{-1}  (\boldsymbol{h}\otimes \qb) \right).
		\end{align*}
		This implies that the $1-\alpha$ CI is statistically valid, since this closed form distribution was used to generate the endpoints of the CI interval in Algorithm \ref{alg:est_ci}.
	\end{proof}

    \section{Time Complexity of Algorithm 1} \label{sec:appendix:time_complexity}

    This section provides the computational cost of Algorithm \ref{alg:two_pass_est}, when the algorithms used for the product of two $n\times n$ matrices and for the matrix inverse have a time complexity of $O(n^3).$

    \begin{theorem}\label{thm:appendix:time_complexity}
        Suppose that for each vertex $v\in\tree,$ $\yb(v) \in\rr^m$ and $\betab(v)\in \rr^n.$ Also, let $V$ be defined as the total number of vertices in $\tree,$ \textit{i.e.}, $V=\sum_l \card(\level(\tree, l)).$ Then the time complexity of Algorithm \ref{alg:two_pass_est} is $O(m^2 n V).$
	\end{theorem}
	\begin{proof}
		The first step in this algorithm is to compute the GLS estimator $\betab(v)$ and its variance matrix based solely on $\yb(v),$ for each vertex $v\in\tree,$ which has a time complexity of $O((m^2 n + n^3) V).$ Since $S(v)$ has full rank for each vertex $v\in\tree,$ we have $m\geq n,$ so this time complexity can be written as, $O(m^2 n V).$ 
        
        Next we will show that the remaining steps in the Algorithm have a time complexity less than $O(m^2 n V).$ First, all of the matrix sums in the fine-to-coarse pass have a time complexity bounded above by $O(n^2 V),$ \textit{i.e.}, the time complexity of adding together the variance matrices from all $V$ vertices. By similar reasoning, the time complexity of all of the matrix inversions in this pass is no more than $O(n^3 V).$ 

        The variance matrix sums, and the inversion operation carried out on this sum afterward, in the coarse-to-fine pass are the same as the ones that were already computed as part of the fine-to-coarse pass, so the operations with the highest time complexity in this pass are products of $n\times n$ matrices. Since the number of these matrix products can be bounded above by a value that is proportional to $V,$ the time complexity of this pass is no more than $O(n^3 V).$
	\end{proof}

	\section{Inverse-Variance Weighted Vectors}  \label{sec:appendix:preliminary_result}
	
	Since the fine-to-coarse recursion uses the inverse-variance weighted mean of two random vectors, the next Lemma provides the formula for the resulting vector and its variance. 
	
	\begin{lemma} \label{lem:inv_var_wtd_mean}
		Suppose $\ab,\bb\in\rr^n$ are realizations of random variables, each with finite variances and with a mean equal to $\psib\in \rr^n.$ Then, the minimum variance unbiased estimate of $\psib$ that is linear in $\ab,\bb$ is given by
		\begin{align*}
			\hat{\psib} = \var(\hat{\psib}) \left( \var(\ab)^{-1}\ab + \var(\bb)^{-1}\bb\right)
		\end{align*}
		\noindent
		where
		\begin{align*}
			\var(\hat{\psib}) = \left( \var(\ab)^{-1} + \var(\bb)^{-1}\right)^{-1}
		\end{align*}
	\end{lemma}
	\begin{proof}
		Let $\boldsymbol{c} = (\ab^\top, \bb^\top)^\top,$ $X = \stack(\{I_n, I_n\}),$ and let the block matrix $\Omega$ be defined as
		
		\begin{align*}
			\Omega = \begin{bmatrix}
				\var(\ab) & 0 \\
				0 & \var(\bb)
			\end{bmatrix}.
		\end{align*}
		
		\noindent
		Now consider applying a GLS estimator to a dataset with the $i^{\textrm{th}}$ independent and dependent variables given by $Y_i$ and $X_{i,\cdot}^\top,$ respectively. The resulting estimator is equal to
		$$\check{\psib} = \left( X^\top \Omega^{-1} X \right)^{-1} X^\top \Omega^{-1} \boldsymbol{c},$$
		\noindent
		which has variance given by 
		$$ \var(\check{\psib}) = \left( X^\top \Omega^{-1} X \right)^{-1}.$$ 
		\noindent
		Also, the Gauss-Markov theorem implies this estimator is the best linear unbiased estimator for $\psib;$ see for example, \citep{greene2003econometric}. The final result follows from 
		\begin{align*}
			\var(\check{\psib}) = \left( \var(\ab)^{-1} + \var(\bb)^{-1}\right)^{-1} = \var(\hat{\psib}),
		\end{align*}
		\noindent
		and
		\begin{align*}
			\check{\psib} = \left( X^\top \Omega^{-1} X \right)^{-1} X^\top \Omega^{-1} \boldsymbol{c} =\var(\hat{\psib}) \left( \var(\ab)^{-1}\ab + \var(\bb)^{-1}\bb\right) = \hat{\psib}.
		\end{align*}
	\end{proof}

\end{document}